\documentclass[letterpaper, onecolumn, unpublished, allowtoday]{quantumarticle}
  \pdfoutput=1

  \usepackage[utf8]{inputenc}
  \usepackage[T1]{fontenc}
  \usepackage[english]{babel}

  \usepackage{xcolor} 

  \newif\ifshowcomments
  \showcommentstrue     

  \newcommand{\cmt}[2][blue]{%
    \ifshowcomments
      {\begingroup\color{#1}\bfseries[\,#2\,]\endgroup}%
    \fi
  }

  \newcommand{\matt}[1]{\cmt[magenta]{Matt: #1}}

  \usepackage{amsmath} 
  \usepackage{amssymb} 
  \usepackage{amsthm}  
  \usepackage{bm}      
  \usepackage{physics} 
  \usepackage{mathtools} 
  \usepackage{verbatim}

  \usepackage{tikzit}

\tikzstyle{rightarrow}=[->]

  \usepackage{graphicx}
  \usepackage{float}
  \usepackage{subcaption}
  \usepackage{caption}
  \usepackage[numbers,sort&compress]{natbib}

  \newcommand{\ga}[1]{\cmt[blue]{GA: #1}}

  \usepackage{xcolor}
  \usepackage{hyperref}
  \hypersetup{
      colorlinks=true,
      linkcolor=blue,
      filecolor=magenta,
      urlcolor=cyan
  }

  \usepackage{etoolbox} 

  \newtheorem{theorem}{Theorem}[section] 
  \newtheorem{lemma}[theorem]{Lemma}     
  \newtheorem{corollary}[theorem]{Corollary}

  \theoremstyle{definition}
  \newtheorem{definition}[theorem]{Definition}

  \theoremstyle{remark}

  \title{Deriving the Generalised Born Rule from First Principles}

  \author{Gaurang Agrawal}
  \affiliation{Indian Institute of Science Education
  and Research Pune, India}
  \affiliation{Université Paris-Saclay, Inria, CNRS, LMF, 91190 Gif-sur-Yvette, France}
  \email{agarwalga20@gmail.com}
  
  \author{Matt Wilson}
  \affiliation{Université Paris-Saclay, CNRS, ENS Paris-Saclay, Inria, CentraleSupélec, Laboratoire Méthodes Formelles}
  \email{matthew.wilson@centralesupelec.fr}

\begin{document}
  \maketitle

  \begin{abstract}
A basic postulate of modern compositional approaches to generalised physical theories is the generalised Born rule, in which probabilities are postulated to be computable from the composition of states and effects. In this paper we consider whether this postulate, and the strength of the identification between scalars and probabilities, can be argued from basic principles. To this end, we first consider the most naive possible process-theoretic interpretation of textbook quantum theory, in which physical processes (unitaries) along with states and effects (kets and bras) and a probability function from states and effects satisfying just some basic compatibility axioms are identified. We then show that any process theory equipped with such structure is equivalent to an alternative process theory in which the generalised Born rule holds. We proceed to consider introduction of noise into any such theory, and observe that the result of doing so is a strengthening of the identification between scalars and probabilities; from bare monoid homomorphisms to semiring isomorphisms.
  \end{abstract}

  \section{Introduction}

Within process theories \cite{Coecke_Kissinger_2017} and operational probabilistic theories \cite{Chiribella2010Jun}, where physical processes are modelled using symmetric monoidal categories (SMCs) \cite{Abramsky2004Feb, Baez2010Jul}, the generalised Born rule stands as a cornerstone principle. It posits a direct connection between the compositional structure of a theory and its observable statistics: the probability of measuring some state $\rho$ and returning result $\sigma$ is given up to homomorphism $\lambda$ by the composition rule \cite{Abramsky2004Feb, Coecke2015Oct, Coecke2017Mar, Chiribella2010Jun} \begin{equation}
  P\left(\input{staterho.tikz}, \input{effectsigma.tikz}\right) =\lambda\left(\input{compositionrhosigma.tikz}\right).
\end{equation}

But how fundamental is this principle? Is its prevalence a mere convenience found in certain well-behaved theories, or is it a derivable feature of any attempt to equip a process theory with a probabilistic interpretation?
If the principle were not derivable, it could well expand the search for generalised compositional and probabilistic theories \cite{Barrett2005Aug, Plavala2023Sep} beyond quantum mechanics, by explicit separation between compositional and probabilistic principles, and highlight the importance of the assumption of the identification of this structure in modern reconstructions of quantum theory. If the principle were derivable it would establish a basic structural fact about the nature of composition and probability in all conceivable physical theories, placing current axiomatic approaches to reconstructing quantum mechanics on firmer ground.

In this article we put forward a derivation of the generalised Born rule from basic principles. We will show that any process theory equipped with a reasonable probabilistic interpretation, which we refer to as \textit{probabilistic process theories} can be systematically transformed into an operationally equivalent one in which the generalised Born rule holds up to an injective monoid homomorphism $\lambda$ from the scalars of the process theory to the monoid of positive Reals under multiplication. The result is first presented for a restricted class of simplified probabilistic process theories as a pedagogical introduction to the main result, and then for the general case.


Building from this main result, we reconsider the precise strength of the identification that can be made between scalars and probabilities. 
Whilst injective monoid endomorphisms on the positive reals are abundant, there exists only one injective semiring endomorphism on the positive reals. 
Inspired by this, we consider whether any probabilistic process theory can be further transformed into an equivalent one in which the identification between scalars and probabilities holds up-to semiring homomorphism, in other words a homomorphism which further preserves sums.
In order to do so, we freely add noise to probabilistic process theories (which satisfy the generalised Born rule) as formal positive linear combinations of processes.
Further quotienting by probabilistic equivalence, gives a way to construct probabilistic process theories in which the generalised Born rule holds up-to an injective semiring homomorphism, see Fig. \ref{fig1}. In the case of noisy quantum theory $\mathbf{CP}$, this allows us to derive the exact quantum mechanical Born rule $\Tr[\rho \sigma]$

  \begin{figure}[H]
      \centering
      \includegraphics[width = 0.6\textwidth]{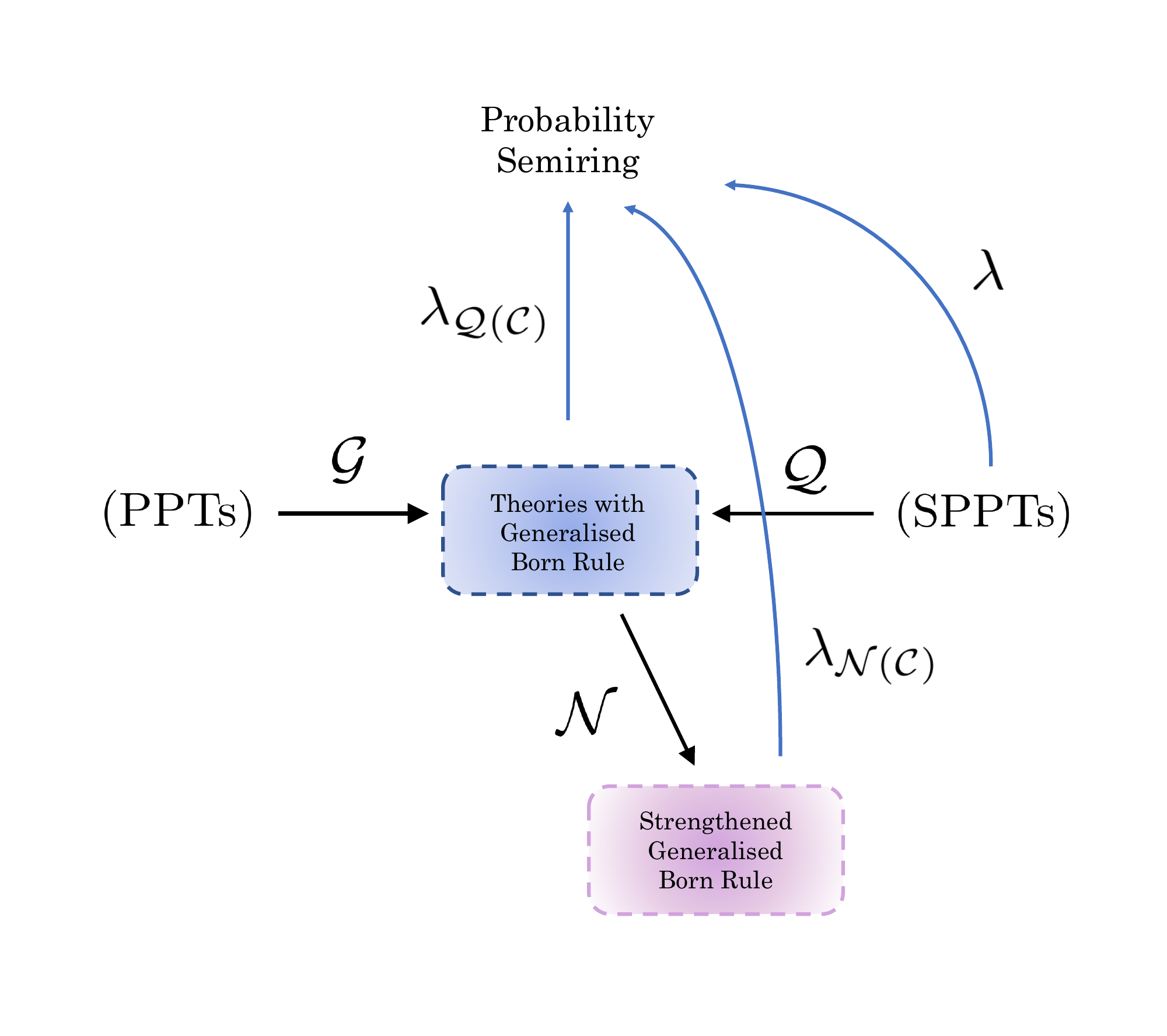}
      \caption{\textbf{Schematic of our construction.} We construct theories with generalised Born rule from probabilistic process theories via a quotient $\mathcal{G}$ ($\mathcal{Q}$ for a simplified case). Further, by adding noise $\mathcal{N}$ we strengthen the generalised Born rule. As a result the function $\lambda$ elevates from a monoid homomorphism to a semiring isomorphism.}
      \label{fig1}
  \end{figure}



     
We explore concrete examples along the way,
for the probabilistic process theory of normalised (pure states and unitaries) quantum theory the construction of an associated theory satisfying the generalised Born rule returns the theory of rank-1 completely positive trace non-increasing (CPTNI) maps i.e., Kraus operators. Adding noise, further yields the theory $\mathbf{CP}$ of completely positive maps. Similarly, for the simplified probabilistic process theory $\mathbf{FHilb}$ of finite-dimensional Hilbert spaces, constructing theories with generalised Born rule and adding noise return the theory of completely positive maps. As a byproduct of this investigation then, we find a categorical construction for the completely positive maps free of adjoints and discards, in contrast with existing constructions such as Selinger's $\mathbf{CP}$ construction \cite{Selinger2007Mar}, the $\mathbf{CP}^*$ 
\cite{Coecke2013May, Coecke2014Aug}, 
and the $\mathbf{CP}^\infty$ construction \cite{Coecke2016Oct} where adjoints are treated as first-class citizens. 

Since this alternative construction for completely positive maps can be applied to any probabilistic process theory, it gives a construction for generalisations of completely positive maps for quantum mechanics with alternative Born rules.
   The obtained theorems could hence also be considered an attempt among a long line of works \cite{Gleason1975, Deutsch1999Aug, Hardy2001Jan, Saunders2004Jun, Carroll2014, Masanes2019Mar, DeBrota2021Aug, Hossenfelder2021Feb, Gogioso2023Jun} to consider the impact of alternative Born rules in quantum theory. 
   In future, for instance, it could be possible to even single out the traditional quantum mechanical Born rule by ruling out pathological features of the associated information theories they induce.
   By the generality of the construction, it furthermore appears that probabilistic process theories could be used as a general method to lift alternative Born rules in quantum mechanics into generalised probabilistic theories \cite{Hardy2001Jan, Barrett2007Mar, Plavala2023Sep} and by the categorical nature of the construction, operational probabilistic theories \cite{Chiribella2010Jun, DAriano2017Jan}.

  \section{Probabilistic Process Theories}

  Graphical notation is used throughout the paper to represent morphisms in SMCs following the conventions of Refs. \cite{Coecke2015Oct, Coecke2017Mar}. Along with being intuitive, the graphical notation is also formal by soundness and completeness of string diagrams for monoidal categories \cite{Joyal1991Jul}. We present a brief overview of symmetric monoidal categories and graphical notation

  \begin{definition}[\bf Symmetric Monoidal Category]
  A symmetric monoidal category (SMC) is a category $\mathcal{C}$ equipped with a bifunctor $\otimes: \mathcal{C} \times \mathcal{C} \rightarrow \mathcal{C}$, a unit object $I$, and natural isomorphisms 
  \begin{align}
    &\alpha_{A,B,C}: (A \otimes B) \otimes C \cong A \otimes (B \otimes C) \qquad &\text{(associator)},\notag\\ 
    &\beta^l_A: I \otimes A \cong A \qquad &\text{(left unitor)}, \notag\\
    &\beta^r_A: A \otimes I \cong A \qquad &\text{(right unitor)}, \notag\\ 
    &\sigma_{A,B}: A \otimes B \cong B \otimes A \qquad &\text{(symmetry)}
  \end{align}
  for all objects $A, B, C$ in $\mathcal{C}$. These isomorphisms must satisfy the coherence conditions, ensuring that all diagrams formed by these natural isomorphisms commute. Furthermore, we can always pretend that these isomorphisms are equalities since, by Mac Lane's strictification theorem \cite{Lane} every SMC is equivalent to a strict one.
  \end{definition}

  The objects of an SMC are drawn as wires with morphisms as boxes connecting to these wires. A sequential composition is represented by connecting the output of one box to the input of another, while a parallel composition is represented by placing boxes side by side, see Fig. \ref{fig:notation}. The monoidal unit object $I$ and morphism $1_I$ are depicted as empty space. Morphisms from $I$ to an object $A$ are called states, while morphisms from an object $A$ to $I$ are called effects. Scalars, which are morphisms from $I$ to $I$, are represented as closed diagrams or simply as numbers in the graphical notation.

    \begin{figure}[H]
      \centering
      \input{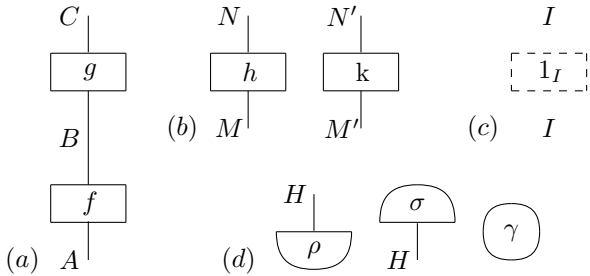}
      \caption{\textbf{Diagrammatic notation.} (a) Sequential composition $g \circ f$ of morphisms $B \overset{g}{\rightarrow} C$ and $A \overset{f}{\rightarrow} B$. (b) Parallel composition $h \otimes k$ of morphisms $M \overset{h}{\rightarrow} N$ and $M' \overset{k}{\rightarrow} N'$. (c) Monoidal unit object $I$ and identity morphism $1_I$ represented as empty diagram (dashed box here). (d) State $\rho$, effect $\sigma$ and scalar $\gamma$ as boxes with single output wire, single input wire, and no input-output wires respectively.}
      \label{fig:notation}
  \end{figure}

  \label{sec:physicaltheories}
  
In modern operational approaches to modelling generalised physical theories it is common to posit a process theory with an assignment from scalars to probabilities. We instead begin with a process-theoretic interpretation of traditional textbook quantum theory. In the more traditional approach, one is presented with notions of
\begin{itemize}
\item States: Normalised elements of a Hilbert space
\item Effects: Normalised Elements of the dual Hilbert space
\item Processes: Unitary Linear maps between Hilbert spaces
\item Probability/Born rule: Assignment from states and effects to Real numbers between $0$ and $1$
\end{itemize}
Probabilistic process theories (PPTs) consists of two parts: a compositional structure, defined by an underlying symmetric monoidal category encoding the states, effects, and processes, along with probability functions satisfying some simple well-behaviour conditions with respect to composition, reminiscent of the conditions for operational probability functions give in Refs.\cite{Galley2017Jul, Galley2018Nov, Masanes2019Mar}.

  \begin{definition}[\bf Probabilistic Process Theory] A probabilistic process theory is a tuple $(\mathcal{D} \subseteq \mathcal{C}, \{\mathcal{S}^A\}, \{\mathcal{E}^A\}, \left\{P^A |~ A \in \mathrm{Ob}(\mathcal{C})\right\})$ consisting of a pair of symmetric monoidal categories $\mathcal{D} \subseteq \mathcal{C}$ along with families of 'physical' states, effects, and probability functions for each object, satisfying the following:
  \begin{itemize}
  \item For any objects $A, B$, the set of physical states $\mathcal{S}^A \subseteq \mathcal{C}(I, A)$ has  \[(s_A,s_B) \in \mathcal{S}^A \times \mathcal{S}^B \implies s_A\otimes s_B \in \mathcal{S}^{A \otimes B}\]
  \item For any objects $A, B$, the set of physical effects $\mathcal{E}^A$ has \[(e_A,e_B) \in \mathcal{E}^A \times \mathcal{E}^B \implies e_A\otimes e_B \in \mathcal{E}^{A \otimes B}\]
  \item For all physical transformations $ f \in \mathcal{D}(A, B),~ \rho \in \mathcal{S}^A, \text{ and } \sigma \in \mathcal{E}^B: ~~ f \circ \rho \in \mathcal{S}^B \text{ and } \sigma \circ f \in \mathcal{E}^A$
  \item The identity on monoidal unit $i_I$ is both a physical state and effect $i_I \in \mathcal{S}^I, ~ i_I \in \mathcal{E}^I$
  \item There exists a probability function $P^A : \mathcal{S}^A \times \mathcal{E}^A \rightarrow \mathbb{R}_{\geq 0}$ 
  \end{itemize}
  subject to three constraints:\begin{itemize}
      \item[\textbf{(I)}] Associativity: \begin{eqnarray}
         P^{B}\left(\input{statefcircrho.tikz} ~ , ~ \input{effectsigma.tikz}\right) &=& P^{A}\left(\input{staterho.tikz} ~ , ~\input{effectsigmacircf.tikz}\right)\notag\\ && \qquad \forall ~ \input{staterho.tikz} \in \mathcal{S}^A,~ \input{effectsigma.tikz} \in \mathcal{E}^B, \text{ and } \input{morphismf.tikz} \in \mathcal{D}(A, B)
      \end{eqnarray} 
      \item[\textbf{(II)}] Product: 
      \begin{eqnarray} P^{A_1\otimes A_2}\left(\input{staterhoone.tikz} \input{staterhotwo.tikz}~,~ \input{effectsigmaone.tikz} \input{effectsigmatwo.tikz}\right) &=& P^{A_1}\left(\input{staterhoone.tikz}~,~ \input{effectsigmaone.tikz}\right) \cdot P^{A_2}\left(\input{staterhotwo.tikz}~,~ \input{effectsigmatwo.tikz}\right) \notag\\ &&\forall ~ \input{staterhoone.tikz} \in \mathcal{S}^{A_1}, ~ \input{staterhotwo.tikz} \in \mathcal{S}^{A_2},~ \input{effectsigmaone.tikz} \in \mathcal{E}^{A_1},~ \input{effectsigmatwo.tikz} \in \mathcal{E}^{A_2} 
       \end{eqnarray} 
      \item[\textbf{(III)}] Non-triviality:  \begin{eqnarray} &\exists& \left(\input{staterho.tikz}, \input{effectsigma.tikz}\right) \text{ such that } P\left(\input{staterho.tikz}, \input{effectsigma.tikz}\right) \neq 0, \notag\\ &\exists& \left(\input{staterhoprime.tikz}, \input{effectsigmaprime.tikz}\right) \text{ such that }  P\left(\input{staterhoprime.tikz}, \input{effectsigmaprime.tikz}\right) \neq 1. \end{eqnarray}      
  \end{itemize}

  \label{def:physicaltheory_new}
  \end{definition}

  \noindent These axioms come with an intuitive set of physical interpretations. Axiom \textbf{(I)} can be attributed to the associativity of physical processes, whereby the calculated probability is independent of whether the transformation $f$ is considered part of the state preparation or the measurement. Axiom \textbf{(II)} asserts that the joint probability of two independent events is the product of their individual probabilities. Lastly, \textbf{(III)} serves to exclude trivial theories where either nothing ever happens (all probabilities are zero) or everything always happens (all probabilities are one). 

  Definition \ref{def:physicaltheory_new} is general enough to encompass pure quantum theory, noisy quantum theory, classical stochastic processes, and many other mathematical frameworks. In particular, the category of finite-dimensional Hilbert spaces, $\mathbf{FHilb}$ with states given by unit-norm elements of Hilbert spaces, effects given by unit-norm elements of dual Hilbert spaces, processes given by the unitary linear maps, and the probability function given by the Born rule, $P(\ket{\psi}, \bra{\phi}) = |\langle \phi | \psi \rangle|^2$, satisfies the definition of a probabilistic process theory. This will be our core example throughout the paper.


Our goal will be to prove that any probabilistic process theory is operationally equivalent to one in which the generalised Born rule holds. Let us state formally what we mean by a generalised Born rule.
  \begin{definition}
  A probabilistic process theory is said to have a generalised Born rule if there exists an injective monoid homomorphism $\lambda$ such that, for any given physical state-effect pair $\rho$ and $\sigma$, the probability of measuring $\sigma$ given $\rho$, $P(\rho, \sigma)$ is exactly the image of their composition under $\lambda$
  \begin{equation}
    P\left(\input{staterho.tikz}, \input{effectsigma.tikz}\right) =\lambda\left(\input{compositionrhosigma.tikz}\right).
  \end{equation}
 Furthermore, we say that this Born rule is additive if $\mathcal{C}$ is enriched in convex spaces and $\lambda$ is a semiring homomorphism. 
  \end{definition}

The proof that any probabilistic process theory is equivalent to one with a generalised Born rule is however rather technical, and so to see the core structures behind the main derivation we will first present a derivation of the same fact for simplified probabilistic process theories (SPPTs), which are process theories imparted with a minimal statistical structure. 
  
    \begin{definition}[\bf Simplified probabilistic process theory]
    A simplified probabilistic process theory is defined as a pair $\left(\mathcal{C}, \left\{P^A |~ A \in \mathrm{Ob}(\mathcal{C})\right\}\right)$ of a symmetric monoidal category $\mathcal{C}$ and a collection of probability functions for all state-effect pairs $P^{A}: \mathcal{S}^A \times \mathcal{E}^A \rightarrow \mathbb{R}_{\geq 0}$, where $\mathcal{S}^A = \mathcal{C}(I, A)$ and $\mathcal{E}^A = \mathcal{C}(A, I)$, subject to three constraints:
  \begin{itemize}
      \item[\textbf{(I)}] Associativity: \( P^{B}(f \circ \rho, \sigma) = P^{A}(\rho, \sigma \circ f) \qquad \forall \rho \in \mathcal{S}^A,~ \sigma \in \mathcal{E}^B, \text{ and } f \in \mathcal{C}(A, B)\) 
      \item[\textbf{(II)}] Product: \( P^{A_1\otimes A_2}(\rho_1 \otimes \rho_2, \sigma_1 \otimes \sigma_2) = P^{A_1}(\rho_1, \sigma_1) \cdot P^{A_2}(\rho_2, \sigma_2) \) 
      \item[\textbf{(III)}] Non-triviality: There exists a state-effect pair $(\rho, \sigma)$ such that \( P(\rho, \sigma) \neq 0 \) and a state-effect pair $(\rho', \sigma')$ such that \( P(\rho', \sigma') \neq 1 \).
  \end{itemize}
  \label{def:physicaltheory}
  \end{definition}
Note, that the trivial case $(\mathcal{C} \subseteq \mathcal{C}, \mathcal{C}(I, A), \mathcal{C}(A,I), \{P^A\})$ within a probabilistic process theory leads to a simplified probabilistic process theory.

  \begin{lemma}
  There exists a state-effect pair $(\rho, \sigma)$ such that \( P(\rho, \sigma) \neq 0 \) if and only if \( P(1_I, 1_I) = 1\), where $1_I$ is the identity on monoidal unit $I$.
  \label{lemma:probabilityunit}
  \end{lemma}

  \begin{proof}
  The forward implication is clear: \( P(1_I, 1_I) = 1 \implies \exists ~ \rho, \sigma \text{ such that } P(\rho, \sigma) \neq 0 \). From \textbf{(II)}, \(P(1_I, 1_I) = P(1_I \otimes 1_I, 1_I \otimes 1_I) = P(1_I, 1_I) \cdot P(1_I, 1_I)\), which implies \(P(1_I, 1_I) \in \{0, 1\}\). However, since $P(\rho, \sigma) = P(1_I, 1_I) \cdot P(\rho, \sigma)$ for all $\rho, \sigma$, the case $P(1_I, 1_I) = 0$ would imply $P(\rho, \sigma) = 0$ for all pairs, which contradicts the premise. Therefore, $P(1_I, 1_I) = 1$.
  \end{proof}

  In theories that possess a discard map, such as noisy quantum theory where the partial trace acts as a discard effect, an additional property holds:
  \begin{lemma}
    In an SMC equipped with a discard map, denoted $\input{discard.tikz},$ $P\left(\input{staterho.tikz}, \input{discard.tikz}\right) = 1$ if and only if $P(1_I, 1_I) = 1$, where $\input{staterho.tikz}$ is a causal state and $1_I$ is the monoidal unit morphism.
    \end{lemma}
    
    \begin{proof}
    $P\left(\input{staterho.tikz}, \input{discard.tikz}\right) \overset{\textbf{(I)}}= P\left(1_I, \input{compositionrhodiscard.tikz}\right) = P(1_I, 1_I)$. The second equality follows from $\rho$ being a causal state. By Lemma \ref{lemma:probabilityunit}, this value is 1.
    \end{proof}

  \section{Probabilities in Simplified Probabilistic Process Theories}
  \label{sec:probphysicaltheories}

In this section we will explore some basic features of SPPTs.

  \begin{lemma}[\bf $\lambda$ function]
  Consider a symmetric monoidal category $\mathcal{C}$ and a probability function $P:\mathcal{S}^A \times \mathcal{E}^A \rightarrow \mathbb{R}_{\geq 0},$ where $\mathcal{S}^A$ and $\mathcal{E}^A$ are the sets of states and effects for object A in category $\mathcal{C}$. Given P satisfies \textbf{(I)}, there exists a function $\lambda:\mathcal{C}(I, I) \rightarrow \mathbb{R}_{\geq 0}$ such that
  \begin{equation}
  \lambda\left(\input{compositionrhosigma.tikz}\right) = P\left(\input{staterho.tikz} \ , \ \input{effectsigma.tikz}\right) \ \forall \rho \in \mathcal{S}^A, \ \sigma \in \mathcal{E}^A.
  \end{equation}
  \label{lemma:lambda}
  \end{lemma}

  \begin{proof}
  $P\left(\input{staterho.tikz} \ , \ \input{effectsigma.tikz}\right) = P\left(\input{staterho.tikz} \circ 1_I \ , \ \input{effectsigma.tikz}\right) \overset{\textbf{(I)}}= P\left(1_I , \ \input{compositionrhosigma.tikz}\right)$. Now we define $\lambda\left(\input{compositionrhosigma.tikz}\right) = P\left(1_I , \ \input{compositionrhosigma.tikz}\right).$ 
  \end{proof}

  \begin{corollary}
  $\input{compositionrhoprimesigmaprime.tikz} = \input{compositionrhosigma.tikz} \Rightarrow P\bigg(\input{staterho.tikz} \ , \ \input{effectsigma.tikz}\bigg) = P\bigg(\input{staterhoprime.tikz} \ , \ \input{effectsigmaprime.tikz}\bigg). $
  \label{corr:compprobequiv}
  \end{corollary}

  This corollary states that if two state-effect pairs compose to the same scalar, they must also yield the same probability. Interestingly, this fact follows directly from the SMC structure and the associativity axiom (\textbf{I}). 

  \begin{theorem}[\bf A monoid homomorphism for simplified probabilistic process theories]
  In a simplified probabilistic process theory $(\mathcal{C}, \{P^A\})$ the function $\lambda$ defined above is a monoid homomorphism $\lambda:\Big(\mathcal{C}(I, I), 1_I, \otimes \Big) \rightarrow \Big(\mathrm{Range}(P) \subseteq \mathbb{R}_{\geq 0}, 1, \times\Big)$.
  \label{th:Cmonoidalhomomorphism}
  \end{theorem}
  \begin{proof}
  From the SPPT $(\mathcal{C}, \{P^A\})$, consider the symmetric monoidal category $\mathcal{C}$ and probability functions $P:\mathcal{S}^A \times \mathcal{E}^A \rightarrow \mathbb{R}_{\geq 0}.$ Given P satisfies \textbf{(I)},\textbf{(II)}, and \textbf{(III)}, we have
  $\lambda(1_I) \overset{\text{Def.}}{=} P(1_I, 1_I) \overset{(\text{Lemma } \ref{lemma:probabilityunit})}=1$. Hence $\lambda$ preserves the identity.\\
  Now, 
  \begin{eqnarray}
  \lambda\left(\input{compositionrhoonesigmaone.tikz} \input{compositionrhotwosigmatwo.tikz}\right) &\overset{\text{Def.}}=& P\left(\input{staterhoone.tikz} \input{staterhotwo.tikz}, \input{effectsigmaone.tikz} \input{effectsigmatwo.tikz}\right)\overset{\textbf{(II)}}=P\left(\input{staterhoone.tikz}, \input{effectsigmaone.tikz}\right)\cdot P\left(\input{staterhotwo.tikz}, \input{effectsigmatwo.tikz}\right) \notag\\ &\overset{\text{Def.}}=&\lambda\left(\input{compositionrhoonesigmaone.tikz}\right) \cdot \lambda\left(\input{compositionrhotwosigmatwo.tikz}\right).
  \end{eqnarray}
  Hence, $\lambda$ preserves scalar multiplication.
  \end{proof}

All three axioms serve their distinct roles. While the associativity axiom \textbf{(I)} gives a well-formed function $\lambda$, the product axiom \textbf{(II)} is responsible for the preservation of the monoidal product, and the non-triviality axiom \textbf{(III)} is responsible for preserving the identity. Interestingly, Lemma \ref{lemma:probabilityunit} shows that identity preservation is equivalent to the theory being non-trivial, i.e., $\lambda(1_I) = 1$ if and only if there exists a state-effect pair such that $P(\rho, \sigma) \neq 0$.

  
  
Homomorphism between scalars and probabilities will not in general be unique. Consider the case of pure quantum theory $\mathbf{FHilb}$, where there is no unique homomorphism from the scalars (complex numbers) to non-negative real numbers $\mathbb{R}_{\geq 0}$. For instance, any power function of the form 
  \begin{equation}
    \lambda(re^{i \theta}) = r^k
  \end{equation}
  is a monoid homomorphism from $(\mathbb{C}, 1, \times)$ to $(\mathbb{R}_{\geq 0}, 1, \times)$ for $k > 0$. The standard Born rule corresponds to the case $k = 2$.

  We now turn to statistical equivalence. In any mathematical model for a given SPPT $(\mathcal{C}, \{P^{A}\})$, multiple descriptions of a transformation may give rise to the same statistics. It is therefore instructive to define statistically equivalent morphisms and indentify their properties. 

  \begin{definition} \label{def:probequivmorphisms}\textbf{(Probabilistically equivalent morphisms)} Given morphisms $f \in \mathcal{C}(A,B)$ and $g \in \mathcal{C}(A,B)$ in SPPT $(\mathcal{C}, \{P^{A}\})$, we call them probabilistically equivalent and write $f \sim g$ iff 
  \begin{equation}
  P^{B \otimes C}\left(\input{statefcircrhodouble.tikz}, \input{effectsigmadouble.tikz}\right) = 
  \lambda\left(\input{compositionrhofsigmadouble.tikz}\right) = \lambda\left(\input{compositionrhogsigmadouble.tikz}\right) = P^{B \otimes C }\left(\input{stategcircrhodouble.tikz}, \input{effectsigmadouble.tikz}\right)
  \end{equation}
  for all $\rho \in \mathcal{S}^{A \otimes C}$, $\sigma \in \mathcal{E}^{B \otimes C}$, and for all systems $C$.
  \end{definition}

  We may also refer to these as statistically equivalent morphisms. It is straightforward to show that $\sim$ is an equivalence relation. The following lemmas hold for probabilistically equivalent states.

  \begin{lemma}
  $\input{compositionrhoa.tikz} = \input{compositionrhoprimea.tikz} ~~ \forall a \in \mathcal{E}^A \implies \input{staterho.tikz} \sim \input{staterhoprime.tikz}.$
  \label{th:localeffectimpliesstateequiv}
  \end{lemma}
  \begin{proof}
  It can be seen that
  \begin{eqnarray}
  \input{staterho.tikz} \sim \input{staterhoprime.tikz} &\overset{\text{Def.}}\iff& P\left(\input{staterhocircmdouble.tikz}, \input{effectsigmadouble.tikz}\right) = P\left(\input{staterhoprimecircmdouble.tikz}, \input{effectsigmadouble.tikz}\right) ~\forall ~ m, \sigma \notag\\
  &\overset{\textbf{(I)}}\iff& P\left(1_I, \input{compositionsigmacircrhotensormdouble.tikz}\right) = P\left(1_I, \input{compositionsigmacircrhoprimetensormdouble.tikz}\right).
  \end{eqnarray}
  From the premise
  \begin{eqnarray}
  \input{compositionsigmacircrhotensormdouble.tikz} = \input{compositionsigmaprimerho.tikz} = \input{compositionrhoprimesigmaprime.tikz} = \input{compositionsigmacircrhoprimetensormdouble.tikz} \notag
  \end{eqnarray}
  which implies $P\left(1_I, \input{compositionsigmacircrhotensormdouble.tikz} \right) = P\left(1_I, \input{compositionsigmacircrhoprimetensormdouble.tikz} \right)$, and therefore $\rho \sim \rho'$.
  \end{proof}

Lemma \ref{th:localeffectimpliesstateequiv} says that if two states give the same compositions with all effects, then they must be probabilistically equivalent. A stronger lemma says that two states give the same statistics with all effects if and only if they are probabilistically equivalent. 

  \begin{lemma} \label{th:stateequiv}
  $\lambda\left(\input{compositionrhoa.tikz}\right) = \lambda\left(\input{compositionrhoprimea.tikz}\right) ~~ \forall a \in \mathcal{E}^A \iff \input{staterho.tikz} \sim \input{staterhoprime.tikz}$
  \end{lemma}
  \begin{proof}
  We know that
  \begin{eqnarray}
  \input{staterho.tikz} \sim \input{staterhoprime.tikz} &\overset{(\text{Lemma } \ref{th:localeffectimpliesstateequiv})}\iff& P\left(1_I, \input{compositionsigmacircrhotensormdouble.tikz}\right) = P\left(1_I, \input{compositionsigmacircrhoprimetensormdouble.tikz}\right) ~ \forall m,\sigma \notag\\
  &\iff& P\left(1_I, \input{compositionsigmaprimerho.tikz}\right) = P\left(1_I, \input{compositionrhoprimesigmaprime.tikz}\right) ~ \forall \sigma' \in \mathcal{E}^A \notag\\
  &\overset{\text{Def.}}\iff&\lambda\left(\input{compositionsigmaprimerho.tikz}\right) = \lambda\left(\input{compositionrhoprimesigmaprime.tikz}\right) ~ \forall \sigma' \in \mathcal{E}^A.
  \end{eqnarray}
  \end{proof}

We now check that scalars factor through probability calculations up-to the homomorphism $\lambda$. 

  \begin{theorem}\label{th:scalermultipletheorem}
  Given $\gamma \in \mathcal{C}(I,I),$ $P\left(\gamma ~ \input{staterho.tikz}, \input{effectsigma.tikz}\right) = \lambda(\gamma)\cdot P\left(\input{staterho.tikz}, \input{effectsigma.tikz}\right)$
  \label{th:scalarfactorthrough}
  \end{theorem}
  \begin{proof}
  \begin{eqnarray}
  P\left(\gamma ~ \input{staterho.tikz}, \input{effectsigma.tikz}\right) &=& P\left(\gamma ~\input{staterho.tikz}, 1_I ~ \input{effectsigma.tikz}\right) \notag\\
  &\overset{\textbf{(II)}}=& P(\gamma, 1_I) \cdot P\left(\input{staterho.tikz}, \input{effectsigma.tikz}\right) \notag\\
  &\overset{\textbf{(I)}}=& P(1_I, \gamma) \cdot P\left(\input{staterho.tikz}, \input{effectsigma.tikz}\right) \notag\\
  &\overset{\text{Def.}}=& \lambda(\gamma) \cdot P\left(\input{staterho.tikz}, \input{effectsigma.tikz}\right).
  \end{eqnarray}
  \end{proof}

  Lastly, scalars are equivalent if and only if their corresponding probabilities are the same.

  \begin{lemma} 
  Let $\gamma, \gamma' \in \mathcal{C}(I,I)$. Then,
  $\gamma \sim \gamma' \iff \lambda(\gamma) = \lambda(\gamma')$.
  \label{th:scalarequiv}
  \end{lemma}
  \begin{proof}
  $\gamma \sim \gamma' \overset{\text{Def.}}\iff \lambda\left(\input{compositiongammascalarsigmarho.tikz}\right) = \lambda\left(\input{compositiongammaprimescalarsigmarho.tikz}\right)$ for all state-effect pairs $\rho, \sigma$. 
  By Lemma \ref{th:scalermultipletheorem}, this is equivalent to 
  \begin{equation}\lambda(\gamma)\lambda\left(\input{compositionrhosigma.tikz}\right) = \lambda(\gamma')\lambda\left(\input{compositionrhosigma.tikz}\right) \quad \forall \rho, \sigma.\end{equation}
  Since, axiom \textbf{(III)} guarantees that there is some pair for which $\lambda\Bigg(\input{compositionrhosigma.tikz}\Bigg) \neq 0$, we may conclude that $\lambda(\gamma) = \lambda(\gamma')$.
  \end{proof}



For a PPT to support a generalised Born rule, we require the monoid homomorphism $\lambda$ to be injective. As a result, we must build a new PPT in which morphisms are quotiented by probabilistic equivalence. This procedure is analogous to ignoring the global phase in pure quantum theory. 

\begin{definition}[\bf Quotiented Category]
We define a new category $\mathcal{Q}(\mathcal{C})$ whose objects are the same as those in $\mathcal{C}$. The morphisms in $\mathcal{Q}(\mathcal{C})$ are the equivalence classes of morphisms from $\mathcal{C}$ under the relation defined in Def.~\ref{def:probequivmorphisms}.
\label{def:quotientedcategory}
\end{definition}

  \begin{definition}[\bf Sequential composition in $\mathcal{Q}(\mathcal{C})$] Given two morphisms $[\phi]_{qc} \in \mathcal{Q}(\mathcal{C})(B,C)$ and $[\chi]_{qc} \in \mathcal{Q}(\mathcal{C})(A,B)$ which are equivalence classes of morphisms $\phi \in \mathcal{C}(B,C)$ and $\chi \in \mathcal{C}(A,B)$ respectively, we define their sequential composition as $[\phi]_{qc} \circ_{qc} [\chi]_{qc} := [\phi \circ \chi]_{qc}$. 
  \end{definition}

  \begin{lemma}[\bf Consistency of $\circ_{qc}$] For the sequential composition to be well-defined, if  $\phi_k \sim \phi_m$ and $\chi_l \sim \chi_n$ we want $ \phi_k \circ \phi_m \sim \chi_l \circ \chi_n.$ This is indeed the case.
  \label{lem:consistencyofcomposition_QC}
  \end{lemma}
  \begin{proof}
 We have

  \begin{equation}
    \phi_k \sim \phi_m \overset{\text{Def.}}{\implies} \lambda\left(\input{compositionrhophiksigmadouble.tikz}\right) = \lambda\left(\input{compositionrhophimsigmadouble.tikz}\right) ~\forall \rho, \sigma,
  \end{equation}

  and
    \begin{equation}
      \chi_l \sim \chi_n \overset{\text{Def.}}{\implies} \lambda\left(\input{compositionrhochilsigmadouble.tikz}\right) = \lambda\left(\input{compositionrhochinsigmadouble.tikz}\right) ~\forall \rho, \sigma.
    \end{equation}

 $\phi_k \circ_{c}\chi_l \sim \phi_m \circ_{c}\chi_l \sim \phi_m \circ_{c}\chi_n$ follows from

    \[\lambda\left(\input{compositionsigmaphikchimrhodoubleextended.tikz}\right) \overset{\phi_k \sim \phi_m}{=} \lambda\left(\input{compositionsigmaphimchimrhodoubleextended.tikz}\right)  \overset{\chi_l \sim \chi_n}{=} \lambda\left(\input{compositionsigmaphimchinrhodouble.tikz}\right).\]
  \end{proof}

  \begin{definition}[\bf Monoidal product in $\mathcal{Q}(\mathcal{C})$] The monoidal product for objects in $\mathcal{Q}(\mathcal{C})$ is the same as in $\mathcal{C}$. For morphisms, given $[\phi]_{qc} \in \mathcal{Q}(\mathcal{C})(A,B)$ and $[\chi]_{qc} \in \mathcal{Q}(\mathcal{C})(C,D)$ which are equivalence classes of morphisms $\phi \in \mathcal{C}(A,B)$ and $\chi \in \mathcal{C}(C,D)$ respectively, we define their monoidal composition as $[\phi]_{qc} \otimes_{qc} [\chi]_{qc} := [\phi \otimes \chi]_{qc}$. 
  \end{definition}

  \begin{lemma}[\bf Consistency of $\otimes_{qc}$]  For the sequential composition to be well-defined, if $\phi_k \sim \phi_m$ and $\chi_l \sim \chi_n$ we want $\phi_k \otimes \phi_m \sim \chi_l \otimes \chi_n.$ This is indeed the case.
  \end{lemma}
  \begin{proof}
By transitivity, it is enough to show $\phi_k \otimes \chi_l \sim \phi_m \otimes \chi_l$ and $\phi_m \otimes \chi_l \sim \phi_m \otimes \chi_n$. The first equivalence holds since $\phi_k \sim \phi_m$:
    \[\lambda\left(\input{compositionrhophikchilsigmatriple.tikz}\right) = \lambda\left(\input{compositionrhophikchilsigmatripleextended.tikz}\right) \overset{\phi_k \sim \phi_m}{=} \lambda\left(\input{compositionrhophimchilsigmatripleextended.tikz}\right) = \lambda\left(\input{compositionrhophimchilsigmatriple.tikz}\right).\]
    The second equivalence holds since $\chi_l \sim \chi_n$:
  \begin{eqnarray}
  \lambda\left(\input{compositionrhophimchilsigmatriple.tikz}\right) &=& \lambda\left(\input{compositionrhoprimechilphimsigmaprimetriple.tikz}\right) \text{ with } \input{effectsigmaprimetriple.tikz} = \input{effectsigmacircswaptriple.tikz}, \text{ and } \input{staterhoprimetriple.tikz} = \input{staterhocircswaptriple.tikz}\notag\\ &\overset{\chi_l \sim \chi_n}{=}&\lambda\left(\input{compositionrhoprimechinphimsigmaprimetriple.tikz}\right) \text{ with same argument as above} \notag\\
  &=& \lambda\left(\input{compositionrhophimchinsigmatriple.tikz}\right) ~~\forall \sigma, \rho.
  \end{eqnarray}
  Therefore, $\phi_k \otimes \chi_l \sim \phi_m \otimes \chi_n$.
  \end{proof}

  \begin{definition}[\bf Identity Morphisms in $\mathcal{Q}(\mathcal{C})$]
  For every object $A$, the identity morphism $1_{A_{qc}} \in \mathcal{Q}(\mathcal{C})(A,A)$ is the equivalence class of $1_{A} \in \mathcal{C}(A,A)$. 
  \end{definition}
  An immediate application of Theorem \ref{th:scalarequiv} reveals that the identity on monoidal unit $1_{I_{qc}}$ consists of all morphisms $\phi$ such that $\phi \sim 1_{I} \iff \lambda(\phi) = \lambda(1_{I})$.

  With all the necessary components defined, the obvious thing to check here is if the category $\mathcal{Q}(\mathcal{C})$ is indeed a symmetric monoidal category. The answer is affirmative. 
  \begin{theorem}
  $\mathcal{Q}(\mathcal{C})$ is a symmetric monoidal category. 
  \label{th:QCisSMC}
  \end{theorem}
  \begin{proof}
  The objects of $\mathcal{Q}(\mathcal{C})$ are the same as in $\mathcal{C}$ and thus satisfy the SMC axioms. For morphisms, the axioms are inherited from $\mathcal{C}$ because the composition and product operations in $\mathcal{Q}(\mathcal{C})$ are well-defined on equivalence classes. For instance, given a morphism $\mathfrak{f} \in \mathcal{Q}(\mathcal{C})(A,B)$ and a representative $f \in \mathfrak{f}$, we know $1_{I} \otimes f \sim f \otimes 1_{I} \sim f$ in $\mathcal{C}$. Since $\mathfrak{f} \otimes 1_{I_{qc}}$, $1_{I_{qc}} \otimes \mathfrak{f}$, and $\mathfrak{f}$ are their respective equivalence classes, we have $\mathfrak{f} \otimes 1_{I_{qc}} = 1_{I_{qc}} \otimes \mathfrak{f} = \mathfrak{f}$. The remaining axioms follow from similar reasoning.
  \end{proof}

  With the quotiented category $\mathcal{Q}(\mathcal{C})$ defined, and shown to be a symmetric monoidal category, we will proceed to show that it forms a valid probabilistic process theory. We begin with defining a family of probability functions on the newly defined category.  

  \begin{definition}[\bf Probability functions in $\mathcal{Q}(\mathcal{C})$]
    Let $\rho_{qc} \in \mathcal{Q}(\mathcal{C})(I, A)$ and $\sigma_{qc} \in \mathcal{Q}(\mathcal{C})(A, I)$ be the states and effects in the quotiented category $\mathcal{Q}(\mathcal{C})$. We define the family of probability functions $P^A_{\mathcal{Q}(\mathcal{C})}:\mathcal{Q}(\mathcal{C})(I, A) \times \mathcal{Q}(\mathcal{C})(A, I) \rightarrow \mathbb{R}_{\geq 0}$ as follows:
    \begin{equation}
    P^A_{\mathcal{Q}(\mathcal{C})}\left(\input{staterhoQC.tikz},\input{effectsigmaQC.tikz}\right) = P\left(\input{staterho.tikz}, \input{effectsigma.tikz}\right) \text{ for }  \rho \in \rho_{qc} \text{ and } \sigma \in \sigma_{qc}.
    \label{eq:probabilityfunctionQC}
    \end{equation}
  \end{definition}

  We must ensure that the probability function $P^A_{\mathcal{Q}(\mathcal{C})}$ is well-defined, i.e., it does not depend on the choice of equivalence-class representatives.

  \begin{lemma}
    The probability functions $P^A_{\mathcal{Q}(\mathcal{C})}$ are well-defined.
  \end{lemma}
  \begin{proof}
    Let $\rho, \rho' \in \rho_{qc}$ and $\sigma, \sigma' \in \sigma_{qc}$, then
    \begin{eqnarray}
    P^A_{\mathcal{Q}(\mathcal{C})}\left(\input{staterhoQC.tikz},\input{effectsigmaQC.tikz}\right) = P\left(\input{staterho.tikz}, \input{effectsigma.tikz}\right) = P\left(1_I, \input{compositionrhosigma.tikz}\right) = \lambda\left(\input{compositionrhosigma.tikz}\right), \notag\\
      P^A_{\mathcal{Q}(\mathcal{C})}\left(\input{staterhoQC.tikz},\input{effectsigmaQC.tikz}\right) = P\left(\input{staterhoprime.tikz}, \input{effectsigmaprime.tikz}\right) = P\left(1_I, \input{compositionrhoprimesigmaprime.tikz}\right) = \lambda\left(\input{compositionrhoprimesigmaprime.tikz}\right). 
    \end{eqnarray}
  As $\rho, \rho' \in \rho_{qc}$ and $\sigma, \sigma' \in \sigma_{qc}$, we have $\rho \sim \rho'$ and $\sigma \sim \sigma'$, which implies $\sigma \circ \rho \sim \sigma' \circ \rho'$ by lemma \ref{lem:consistencyofcomposition_QC}. By lemma \ref{th:scalarequiv}, we have
  \begin{equation} \sigma \circ \rho \sim \sigma' \circ \rho' \iff \lambda\left(\input{compositionrhosigma.tikz}\right) = \lambda\left(\input{compositionrhoprimesigmaprime.tikz}\right),\end{equation}
which shows that the value of $P^A_{\mathcal{Q}(\mathcal{C})}$ is independent of the choice of representatives from the equivalence classes. 
  \end{proof}

  We can define the $\lambda$ function for the quotiented category analogously. 

  \begin{definition}
    $\lambda$ function for the quotiented category $\lambda_{\mathcal{Q}(\mathcal{C})}:\mathcal{Q}(\mathcal{C})(I,I) \rightarrow \mathbb{R}_{\geq 0}$ is defined as follows:
    \begin{equation}
    \lambda_{\mathcal{Q}(\mathcal{C})}(\gamma_{qc}) = P^A_{\mathcal{Q}(\mathcal{C})}(1_I, \gamma_{qc}) = P(1_I, \gamma) \text{ for } \gamma \in \gamma_{qc}.
    \end{equation}
  \end{definition}

  As the next step, we show that the quotiented category $\mathcal{Q}(\mathcal{C})$ with the family of probability function $P^A_{\mathcal{Q}(\mathcal{C})}$ constitutes a valid probabilistic process theory.

  \begin{theorem}
  The probability function $P_{\mathcal{Q}(\mathcal{C})}$ satisfies the axioms \textbf{(I)}, \textbf{(II)}, and \textbf{(III)} in the category $\mathcal{Q}(\mathcal{C})$.
  \label{th:QCProbaxioms}
  \end{theorem}

  \begin{proof}
    \textbf{(I)}: Let $\rho_{qc} \in \mathcal{Q}(\mathcal{C})(I,A)$, $\sigma_{qc} \in \mathcal{Q}(\mathcal{C})(B,I)$ and $f_{qc} \in \mathcal{Q}(\mathcal{C})(A,B)$ be the states, effects and morphisms in the quotiented category $\mathcal{Q}(\mathcal{C})$ and let $\rho \in \rho_{qc}, \sigma  \in \sigma_{qc}$ and $f \in f_{qc}$ . Then,
    \begin{equation}
    P_{\mathcal{Q}(\mathcal{C})}\left(\input{statefcircrhodoubleQC.tikz},\input{effectsigmadoubleQC.tikz}\right) = P\left(\input{statefcircrhodouble.tikz}, \input{effectsigmadouble.tikz}\right) = P\left(\input{staterhodouble.tikz}, \input{effectsigmacircfdouble.tikz} \right)  = P_{\mathcal{Q}(\mathcal{C})}\left(\input{staterhodoubleQC.tikz}, \input{effectsigmacircfdoubleQC.tikz}\right).
    \end{equation}
    \textbf{(II)}:  Let $\rho_{qc}, \rho'_{qc} \in \mathcal{Q}(\mathcal{C})(I,A)$ and  $\sigma_{qc}, \sigma'_{qc} \in \mathcal{Q}(\mathcal{C})(B,I)$ be states and effects in the quotiented category $\mathcal{Q}(\mathcal{C})$ and let $\rho \in \rho_{qc}, \rho' \in \rho'_{qc}, \sigma \in \sigma_{qc}, \sigma' \in \sigma'_{qc}$. Then
     \begin{equation}
\begin{aligned}
P_{\mathcal{Q}(\mathcal{C})}\left( \input{staterhoQC.tikz} \input{staterhoprimeQC.tikz},\input{effectsigmaQC.tikz} \input{effectsigmaprimeQC.tikz}\right)
    &= P\left( \input{staterho.tikz} \input{staterhoprime.tikz}, \input{effectsigma.tikz} \input{effectsigmaprime.tikz}\right) \\[-4pt]
    & \hspace*{-1cm} = P\left(\input{staterho.tikz}, \input{effectsigma.tikz}\right) \,  
       P\left(\input{staterhoprime.tikz}, \input{effectsigmaprime.tikz}\right)
       = P_{\mathcal{Q}(\mathcal{C})}\left( \input{staterhoQC.tikz}, \input{effectsigmaQC.tikz}\right)
         P_{\mathcal{Q}(\mathcal{C})}\left( \input{staterhoprimeQC.tikz}, \input{effectsigmaprimeQC.tikz}\right).
\end{aligned}
\end{equation}
    \textbf{(III)}: Since, there exists a pair $(\rho, \sigma)$ in the original theory such that $P(\rho, \sigma) \neq 0$, we can choose $\rho_{qc} \in \mathcal{Q}(\mathcal{C})(I,A)$ and $\sigma_{qc} \in \mathcal{Q}(\mathcal{C})(B,I)$ to be the equivalence classes of $\rho$ and $\sigma$ respectively. Then, we have
    \[
    P_{\mathcal{Q}(\mathcal{C})}( \rho_{qc},\sigma_{qc}) = P(\rho, \sigma) \neq 0. 
    \]
    Similarly, there exists a pair $(\rho', \sigma')$ in the original theory such that $P(\rho', \sigma') \neq 1$, we can choose $\rho'_{qc} \in \mathcal{Q}(\mathcal{C})(I,A)$ and $\sigma'_{qc} \in \mathcal{Q}(\mathcal{C})(B,I)$ to be the equivalence classes of $\rho'$ and $\sigma'$ respectively. 
    \[
    P_{\mathcal{Q}(\mathcal{C})}( \rho'_{qc},\sigma'_{qc}) = P(\rho', \sigma') \neq 1. 
    \]
    
  \end{proof}

  \begin{corollary}
    $(\mathcal{Q}(\mathcal{C}), \{P^A_{\mathcal{Q}(\mathcal{C})}\})$ is a probabilistic process theory.
    \label{cor:QCvalidtheory}
  \end{corollary}

  We term $(\mathcal{Q}(\mathcal{C}), \{P^A_{\mathcal{Q}(\mathcal{C})}\})$ a quotiented SPPT. Now, we examine the impact of this quotienting procedure on pure quantum theory $(\mathbf{FHilb }, |\langle -|- \rangle|^2)$ and noisy quantum theory $(\mathbf{CP}, \Tr[- \circ -])$.

  \begin{theorem}[\bf Characterisation of states and morphisms in quotiented $(\mathbf{FHilb}, |\langle -|- \rangle|^2)$] Starting from $(\mathbf{FHilb }, |\langle -|- \rangle|^2)$ with states  $\ket{\psi} \in \mathbf{FHilb}(I,A)$ and morphisms $f \in \mathbf{FHilb}(A,B)$,  the states and morphisms in quotiented theory  are equivalence classes $\{e^{i\theta}\ket{\psi} | ~\theta \in [0,2\pi)\}$ and $\{e^{i\theta}f | ~\theta \in [0,2\pi)\}$ respectively, i.e., they correspond to unit rays of states and morphisms. Consequently, the process theory produced by quotienting $\mathbf{FHilb}$ is in bijection with the rank-1 completely positive maps of form $\ketbra{\psi}$ and $f (\cdot) f^\dagger$ respectively. 
  \end{theorem}
  \begin{proof}
  Given $\ket{\psi} ,\ket{\psi'} $ are states in $\mathbf{FHilb}$
  \begin{equation}
    \ket{\psi} \sim \ket{\psi'} \overset{(\text{Lemma } \ref{th:stateequiv})}\iff \lambda\left(\input{compositionpsisigma.tikz}\right) = \lambda\left(\input{compositionpsiprimesigma.tikz}\right) ~ \forall \sigma \in \mathcal{S}^B.
    \label{eq:stateequivFHilb}
  \end{equation}
  In $(\mathbf{FHilb }, |\langle -|- \rangle|^2)$ this corresponds to
  \begin{eqnarray}
  && |\langle \sigma| \psi \rangle|^2 = |\langle \sigma| \psi' \rangle|^2 \qquad \forall ~\sigma \notag\\
  &\iff&  \langle \sigma| \psi \rangle = e^{i \theta} \langle \sigma| \psi' \rangle  \notag\\
  &\iff& \langle \sigma |( \ket{\psi} - e^{i \theta} \ket{\psi'}) = 0 \qquad \forall ~\sigma \notag\\
  &\iff& \ket{\psi} = e^{i \theta} \ket{\psi'}.
  \end{eqnarray}
$\therefore$ states in quotiented $(\mathbf{FHilb }, |\langle -|- \rangle|^2)$  are sets of states $e^{i \theta}\ket{\psi}$. Alternatively, note that  \begin{eqnarray}
  & & |\langle \sigma| \psi \rangle|^2 = |\langle \sigma| \psi' \rangle|^2 \qquad \forall ~\sigma \notag\\
  &\iff&  \langle \sigma| \psi \rangle \langle \psi| \sigma \rangle =  \langle \sigma| \psi' \rangle \langle \psi'| \sigma \rangle  \notag\\ 
  &\iff& \langle \sigma |( \ketbra{\psi} - \ketbra{\psi'}) | \sigma \rangle = 0 \qquad \forall ~\sigma \notag\\
  &\iff& \ketbra{\psi} = \ketbra{\psi'}.
  \end{eqnarray}
 $\therefore$ states in quotiented $(\mathbf{FHilb }, |\langle -|- \rangle|^2)$ correspond to rank-1 unnormalised density matrices $\ketbra{\psi}$. 
 For morphisms $f, f'$ in $\mathbf{FHilb}$ such that $f \sim f'$, we have
  \begin{eqnarray}
  f \sim  f' &\iff& \lambda\left(\input{compositionrhofsigmadouble.tikz}\right) = \lambda\left(\input{compositionrhofprimesigmadouble.tikz}\right) ~ \forall \rho, \sigma. 
  \end{eqnarray}
In $(\mathbf{FHilb }, |\langle -|- \rangle|^2)$ this corresponds to
  \begin{eqnarray}
  && |\langle \sigma | f \otimes 1_B| \rho \rangle|^2 = | \langle \sigma | f' \otimes 1_B| \rho \rangle|^2 \qquad \forall \rho, \sigma \notag\\
  &\iff& \langle \sigma | f \otimes 1_B| \rho \rangle = e^{i \theta} \langle \sigma | f' \otimes 1_B| \rho \rangle  \notag\\
  &\iff& \langle \sigma | (f - e^{i \theta}f') \otimes1_B| \rho \rangle = 0 \qquad \forall \rho, \sigma \notag\\
  &\iff& f = e^{i \theta} f' .
  \end{eqnarray}
 $\therefore$ morphisms in quotiented $\mathbf{FHilb}$ are sets of morphisms $e^{i \theta}f$. Alternatively, note that
    \begin{eqnarray}
  && |\langle \sigma | f \otimes1_B| \rho \rangle|^2 = | \langle \sigma | f' \otimes1_B| \rho \rangle|^2 \qquad \forall \rho, \sigma \notag\\
  &\iff& \langle \sigma | f \otimes1_B| \rho \rangle \langle \rho | f^\dagger \otimes1_B| \sigma \rangle = \langle \sigma | f' \otimes1_B| \rho \rangle \langle \rho | f'^\dagger \otimes1_B| \sigma \rangle  \notag\\
  &\iff& \langle \sigma | \left(f \otimes1_B| \rho \rangle \langle \rho | f^\dagger \otimes1_B- f' \otimes1_B| \rho \rangle \langle \rho | f'^\dagger \otimes1_B\right) | \sigma \rangle = 0 \qquad \forall \rho, \sigma \notag\\
  &\iff& f \otimes1_B| \rho \rangle \langle \rho | f^\dagger \otimes1_B = f' \otimes1_B| \rho \rangle \langle \rho | f'^\dagger \otimes1_B\qquad \forall \rho \notag\\
  &\iff& f \text{ and } f' \text{ have the same Choi matrices} \notag\\
  &\iff& \text{they act identically on density matrices } f (\cdot) f^\dagger = f' (\cdot) f'^\dagger.
  \end{eqnarray}
  $\therefore$ morphisms in quotiented $\mathbf{FHilb}$ correspond to rank-1 completely positive map $f (\cdot) f^\dagger$. 
  \end{proof}

  
Note that the probability functions in a quotiented SPPT constructed from pure quantum theory $\{\mathbf{FHilb}, |\langle - | - \rangle|^2\}$ are given by the density matrix Born rule  $P_{\mathcal{Q}(\mathcal{C})}(\ketbra{\psi}{\psi}, \ketbra{\phi}{\phi}) = P(\ket{\psi}, \bra{\phi}) = |\langle \psi | \phi \rangle|^2=  Tr(\ketbra{\psi} \circ \ketbra{\phi})$. 

  \begin{theorem}
  \textbf{Characterisation of states and morphisms in quotiented $(\mathbf{CP}, \Tr[- \circ -])$.} If our original theory is $(\mathbf{CP}, \Tr[- \circ -])$ with morphisms $\phi \in \mathbf{CP}(A,B)$,  the morphisms in quotiented $\mathbf{CP}$ are exactly the same as that of $\mathbf{CP}$, i.e., the theory remains unchanged.
  \end{theorem}
  \begin{proof}
    We know that $f \sim  f' \iff \lambda\left(\input{compositionrhofsigmadouble.tikz}\right) = \lambda\left(\input{compositionrhofprimesigmadouble.tikz}\right) ~ \forall \rho, \sigma.$ In $(\mathbf{CP}, \Tr[- \circ -])$ this corresponds to 
  \begin{eqnarray}
  &\iff& \Tr[\sigma (f \otimes I) \rho] = \Tr[\sigma (f' \otimes I) \rho] \qquad \forall \rho, \sigma \notag\\
  &\iff& \Tr[\sigma ((f - f') \otimes I) \rho] = 0 \qquad \forall \rho, \sigma \notag\\ 
  & \iff& f = f'.
  \end{eqnarray}
  \end{proof}

Lastly, we show that quotienting returns a theory with a generalised Born rule.

  \begin{theorem}[\bf Generalised Born rule for simplified probabilistic process theories.]
    There exists a monoid homomorphism $\lambda_{\mathcal{Q}(\mathcal{C})}$ from $\mathcal{Q}(\mathcal{C})(I, I)$ to $(\mathbb{R}_{\geq 0}, 1, \times)$ and a monoid homomorphism $\theta_{\mathcal{Q}(\mathcal{C})}$ from $(Range(P) \subseteq \mathbb{R}_{\geq 0}, 1, \times)$ to $\mathcal{Q}(\mathcal{C})(I, I)$ such that $\lambda_{\mathcal{Q}(\mathcal{C})}(\theta_{\mathcal{Q}(\mathcal{C})}(a)) = a ~~ \forall a \in Range(P)$ and $\theta_{\mathcal{Q}(\mathcal{C})} (\lambda_{\mathcal{Q}(\mathcal{C})}(\omega)) = \omega ~~\forall \omega \in \mathcal{Q}(\mathcal{C})(I,I)$.
    \label{thm:generalisedBornQC}
  \end{theorem}
  \begin{proof}
That $\lambda_{\mathcal{Q}(\mathcal{C})}$ is a monoid homomorphism follows from $(\mathcal{Q}(\mathcal{C}), \{P^A_{\mathcal{Q}(\mathcal{C})}\})$ being a simplified probabilistic process theory as shown in corollary \ref{cor:QCvalidtheory} and Theorem \ref{th:Cmonoidalhomomorphism}. For the monoid homomorphism $\theta_{\mathcal{Q}(\mathcal{C})}$ consider a pair of scalars $\omega, \omega' \in \mathcal{Q}(\mathcal{C})(I, I)$ such that $\lambda_{\mathcal{Q}(\mathcal{C})}(\omega) = \lambda_{\mathcal{Q}(\mathcal{C})}(\omega'),$ and $\gamma \in \omega$, $\gamma' \in \omega'$, are pair of scalars in $\mathcal{C}$ with $\gamma = \sigma \circ \rho$ and $\gamma' = \sigma' \circ \rho'$.
  Now, since
  \begin{eqnarray}
  \lambda_{\mathcal{Q}(\mathcal{C})}(\omega) = \lambda_{\mathcal{Q}(\mathcal{C})}(\omega') \overset{\text{Def.}}{\iff} \lambda(\gamma) = \lambda(\gamma') 
  \overset{\text{Lemma }\ref{th:scalarequiv}}\iff \gamma \sim \gamma',
  \end{eqnarray}
  which implies $\omega = [\gamma]_{qc} = [\gamma']_{qc} = \omega'$ and so $\lambda_{\mathcal{Q}(\mathcal{C})}$ is injective. Consequently, there exists an inverse $\theta_{\mathcal{Q}(\mathcal{C})}$ such that $\theta_{\mathcal{Q}(\mathcal{C})}(\lambda_{\mathcal{Q}(\mathcal{C})}(\omega)) = \omega$ for all $\omega \in \mathcal{Q}(\mathcal{C})(I, I)$ and $\lambda_{\mathcal{Q}(\mathcal{C})}(\theta_{\mathcal{Q}(\mathcal{C})}(p)) = p$ for all $p \in Range(P_{\mathcal{Q}(\mathcal{C})}) \subseteq \mathbb{R}_{\geq 0}$. Furthermore, 
  \begin{equation}\theta_{\mathcal{Q}(\mathcal{C})}(1) = \theta_{\mathcal{Q}(\mathcal{C})}(\lambda_{\mathcal{Q}(\mathcal{C})}(1_{I_{qc}})) = 1_{I_{qc}}.\end{equation}
  \begin{eqnarray}
    \theta_{\mathcal{Q}(\mathcal{C})}(a \cdot b) = \theta_{\mathcal{Q}(\mathcal{C})}(\lambda_{\mathcal{Q}(\mathcal{C})}(\omega_a) \cdot \lambda_{\mathcal{Q}(\mathcal{C})}(\omega_b)) = \theta_{\mathcal{Q}(\mathcal{C})}(\lambda_{\mathcal{Q}(\mathcal{C})}(\omega_a \otimes \omega_b))  = \omega_a \otimes \omega_b = \theta_{\mathcal{Q}(\mathcal{C})}(a) \otimes \theta_{\mathcal{Q}(\mathcal{C})}(b).
  \end{eqnarray}
  and so $\theta_{\mathcal{Q}(\mathcal{C})}$ is a monoid homomorphism. 
  \end{proof}

  We conclude that for any simplified probabilistic process theory $(\mathcal{C}, \{P^A\})$ the quotiented theory $(\mathcal{Q}(\mathcal{C}), \{P^A_{\mathcal{Q}(\mathcal{C})}\})$ has a generalised Born rule.
  For a given state-effect pair $\rho_{qc}$ and $\sigma_{qc}$, the probability of measuring $\sigma_{qc}$ given $\rho_{qc}$, $P_{\mathcal{Q}(\mathcal{C})}(\rho_{qc}, \sigma_{qc})$ is the image of their composition under the monoid isomorphism $\lambda_{\mathcal{Q}(\mathcal{C})}$
\begin{equation}
  P_{\mathcal{Q}(\mathcal{C})}\left(\input{staterhoQC.tikz}, \input{effectsigmaQC.tikz}\right) =P_{\mathcal{Q}(\mathcal{C})}\left(1_{I_{qc}}, \input{compositionrhosigmaQC.tikz}\right) = \lambda_{\mathcal{Q}(\mathcal{C})}\left(\input{compositionrhosigmaQC.tikz}\right),
\end{equation} 
and conversely, the scalar obtained from their composition is exactly the image of the probability value under the monoid isomorphism $\theta_{\mathcal{Q}(\mathcal{C})}$
  \begin{equation}
    \theta_{\mathcal{Q}(\mathcal{C})}\left(P_{\mathcal{Q}(\mathcal{C})}\left(\input{staterhoQC.tikz}, \input{effectsigmaQC.tikz}\right)\right) = \theta_{\mathcal{Q}(\mathcal{C})}\left(\lambda_{\mathcal{Q}(\mathcal{C})}\left(\input{compositionrhosigmaQC.tikz}\right)\right) = \input{compositionrhosigmaQC.tikz}.
  \end{equation}

Having established the existence of the generalised Born rule in the simplified case, we now turn to the more principled and more complicated problem of proving the generalised Born rule for completely general probabilistic process theories.

  \section{Derivation of the Generalised Born Rule for Probabilistic Process Theories}
  \label{sec:subtheories}

Recall that simplified probabilistic process theories
allow us to define a lambda function $P(\rho, \sigma) = P(1_I, \sigma \circ \rho) = \lambda(\sigma \circ \rho)$, by essentially observing that states are particular cases of processes. A general probabilistic process theory decouples the roles of physical states and morphisms. 
Physical morphisms come from a subcategory $\mathcal{D}$ and the set of physical states and effects $\mathcal{S}^A, \mathcal{E}^A$ may or may not lie in $\mathcal{D}(I, A) \text{ and } \mathcal{D}(A, I)$ respectively. As a consequence, there may not be a well-formed $\lambda$ function. At this expense, however, we can accommodate constructions where the physical entities do not form a category.
Natural examples of theories fit this broader mould, for instance textbook quantum mechanics with unitaries $\mathcal{D}$ and unit-norm kets and bras $\mathcal{S}^A, \mathcal{E}^A$, or textbook quantum information theory with quantum channels $\mathcal{D},$ trace-$1$ density matrices $\mathcal{S}^A$, and general POVM elements $\mathcal{E}^A$.

In this general case the $\lambda_{\mathcal{Q}(\mathcal{C})}$ homomorphism used to express the generalised Born rule will only arise after taking a quotient by probabilistic equivalence. This however presents a technical challenge since we directly referred to the function $\lambda$ to take such quotients. 
In the general case then; we must refer only directly to the probability function and make repeated direct use of the three laws for such functions on probabilistic process theories.


  \begin{definition}[\bf Quotienting a probabilistic process theory]
  Given a probabilistic process theory $(\mathcal{D} \subseteq \mathcal{C}, \{\mathcal{S}^A\}, \{\mathcal{E}^A\}, \{P^A\}),$ we define its associated quotiented process theory $\mathcal{G}(\mathcal{D})$ as follows: 
  \begin{itemize}
  \item Objects: same as $\mathcal{D}$
  \item Morphisms: $\mathcal{G}(\mathcal{D})(A,B) := \{  (U,\rho,\sigma)_{\cong}   \}, ~ U \in \mathcal{D}(A \otimes X, B \otimes X'), \rho \in \mathcal{S}^X, \sigma \in \mathcal{E}^{X'}$
  Diagrammatically, a morphism is an equivalence class of a tuple \[    \left( \input{sub_thm_1.tikz},\input{sub_thm_2.tikz},\input{sub_thm_3.tikz} \right) _{\cong}, \]
  where we say that $ (U_1,\rho_1,\sigma_1) \cong (U_2,\rho_2,\sigma_2)$ if and only if
  \[ \forall \tau \in \mathcal{S}^{A \otimes Z}, \mu \in \mathcal{E}^{B \otimes Z}: \quad   P \left( \input{sub_defn_1a.tikz}, \input{sub_defn_1b.tikz} \right) =  P \left( \input{sub_defn_2a.tikz}, \input{sub_defn_2b.tikz} \right).  \]
  \item Sequential composition: 
   \[   (U,\rho^u,\sigma^u)_{\cong} \circ  (V,\rho^v,\sigma^v)_{\cong} \ := \  \left( \input{sub_thm_comp_1.tikz},\input{sub_thm_comp_2.tikz},\input{sub_thm_comp_3.tikz} \right) _{\cong}    \]
  \item Parallel composition:
  \[   (U,\rho^u,\sigma^u)_{\cong} \otimes  (V,\rho^v,\sigma^v)_{\cong} \ := \  \left( \input{sub_thm_par_comp_1.tikz},\input{sub_thm_comp_2.tikz},\input{sub_thm_comp_3.tikz} \right) _{\cong}    \]
  \end{itemize}
  \label{def:subtheory}
  \end{definition}

Following the previous section \ref{sec:probphysicaltheories}, we now check the consistency of the construction to show that it indeed returns a new PPT which furthermore admits a generalised Born rule. 

\begin{lemma}[\bf Consistency of sequential composition]
  Sequential composition in $\mathcal{G}(\mathcal{D})$ is well-defined.
  \label{lem:subtheory_seq_comp_welldefined}  
\end{lemma}
\begin{proof}
  We need to show that the composition is independent of the choice of representatives from the equivalence classes $(U,\rho^u,\sigma^u)_{\cong},  (V,\rho^v,\sigma^v)_{\cong}$. Consider two pairs of tuples
  \begin{eqnarray} 
    (U_1,\rho_1^u,\sigma_1^u), (U_2,\rho_2^u,\sigma_2^u) \in (U,\rho^u,\sigma^u)_{\cong} 
    \label{eq:subtheory_inclusion1}\\
    (V_1,\rho_1^v,\sigma_1^v), (V_2,\rho_2^v,\sigma_2^v) \in (V,\rho^v,\sigma^v)_{\cong}
    \label{eq:subtheory_inclusion2}
  \end{eqnarray}
  we show that the compositions resulting from any two pairs of representatives are probabilistically equivalent.
\begin{align*}
\forall \tau, \mu: \quad & P \left( \input{sub_comb_seq_proof_1a.tikz}, \input{sub_comb_seq_proof_1b.tikz} \right) &&\overset{(I)}{=} && P \left( \input{sub_comb_seq_proof_2a.tikz}, \input{sub_comb_seq_proof_2b.tikz} \right) 
\end{align*}
\begin{align}
=& P \left( \input{sub_comb_seq_proof_2aa.tikz}, \input{sub_comb_seq_proof_2bb.tikz} \right) &&\overset{Eq.~\eqref{eq:subtheory_inclusion2}}{=} && P \left( \input{sub_comb_seq_proof_3a.tikz}, \input{sub_comb_seq_proof_3b.tikz} \right)\notag \\
=& P \left( \input{sub_comb_seq_proof_3aa.tikz}, \input{sub_comb_seq_proof_3b.tikz} \right) &&\overset{(I)}{=} && P \left( \input{sub_comb_seq_proof_4a.tikz}, \input{sub_comb_seq_proof_4b.tikz} \right) \notag \\
=& P \left( \input{sub_comb_seq_proof_5a.tikz}, \input{sub_comb_seq_proof_5b.tikz} \right) &&\overset{Eq.~\eqref{eq:subtheory_inclusion1}}{=} && P \left( \input{sub_comb_seq_proof_6a.tikz}, \input{sub_comb_seq_proof_6b.tikz} \right) \notag \\
=& P \left( \input{sub_comb_seq_proof_6a.tikz}, \input{sub_comb_seq_proof_6bb.tikz} \right) &&\overset{(I)}{=} && P \left( \input{sub_comb_seq_proof_7a.tikz}, \input{sub_comb_seq_proof_7b.tikz} \right)
\end{align}
\end{proof}

\begin{lemma}[\bf Consistency of parallel composition]
  Parallel composition in $\mathcal{G}(\mathcal{D})$ is well-defined.
  \label{lem:subtheory_par_comp_welldefined}  
\end{lemma}
\begin{proof}
  We need to show that different representations of equivalence classes $(U,\rho^u,\sigma^u)_{\cong},  (V,\rho^v,\sigma^v)_{\cong}$ give rise to the same parallel composition.  Consider two pairs of tuples
  \begin{eqnarray} 
(U_1,\rho_1^u,\sigma_1^u), (U_2,\rho_2^u,\sigma_2^u) \in (U,\rho^u,\sigma^u)_{\cong} 
    \label{eq:subtheory_inclusion3}\\
    (V_1,\rho_1^v,\sigma_1^v), (V_2,\rho_2^v,\sigma_2^v) \in (V,\rho^v,\sigma^v)_{\cong}
    \label{eq:subtheory_inclusion4}
  \end{eqnarray}
   compositions resulting from any two pairs of representatives are probabilistically equivalent. $\forall \tau, \mu:$ 
\begin{align}
 & P \left( \input{sub_comb_par_proof_1a.tikz}, \input{sub_comb_par_proof_1b.tikz} \right) &&\overset{(I)}{=} && P \left( \input{sub_comb_par_proof_2a.tikz}, \input{sub_comb_par_proof_2bb.tikz} \right) \notag\\
=& P \left( \input{sub_comb_par_proof_2aa.tikz}, \input{sub_comb_par_proof_2b.tikz} \right) &&\overset{Eq. \eqref{eq:subtheory_inclusion4}}{=} && P \left( \input{sub_comb_par_proof_3a.tikz}, \input{sub_comb_par_proof_3b.tikz} \right) \notag \\
=& P \left( \input{sub_comb_par_proof_3aa.tikz}, \input{sub_comb_par_proof_3b.tikz} \right) &&\overset{(I)}{=} &&P \left( \input{sub_comb_par_proof_4a.tikz}, \input{sub_comb_par_proof_4b.tikz} \right) \notag \\
& \input{effectmuprimedouble.tikz} = \input{effectmucircswapdouble.tikz}, \input{statetauprimedouble.tikz} = \input{statetaucircswapdouble.tikz} &&&& \text{and we use \textbf{(I)} for } f = \input{swapswap.tikz} \notag  \\
=& P \left( \input{sub_comb_par_proof_5a.tikz}, \input{sub_comb_par_proof_5b.tikz} \right) &&\overset{Eq. \eqref{eq:subtheory_inclusion3}}{=} && P \left( \input{sub_comb_par_proof_6a.tikz}, \input{sub_comb_par_proof_6b.tikz} \right) \notag \\
=& P \left( \input{sub_comb_par_proof_6aa.tikz}, \input{sub_comb_par_proof_6bb.tikz} \right) &&\overset{(I)}{=} && P \left( \input{sub_comb_par_proof_7a.tikz}, \input{sub_comb_par_proof_7b.tikz} \right)
\end{align}
\end{proof}

With the compositions being consistent, we prove that the quotiented theory is an SMC. 

\begin{theorem}[\bf Quotienting returns an SMC]
  For any probabilistic process theory $ \mathcal{G}(\mathcal{D})$ is a symmetric monoidal category.
  \label{th:subtheory_ismonoidal}           
\end{theorem}
\begin{proof} 
  We begin with showing that $\mathcal{G}(\mathcal{D})$ is indeed a category. For this we need to show that sequential composition is associative and that there exists an identity morphism for each object. Associativity follows from
 \begin{align}  
 &(W,\rho^w,\sigma^w)_{\cong} \circ( (V,\rho^v,\sigma^v)_{\cong} \circ  (U,\rho^u,\sigma^u)_{\cong} ) \notag \\
  =& \  \left( \input{sub_seq_ass_1.tikz},\input{sub_seq_ass_2.tikz},\input{sub_seq_ass_3.tikz} \right) _{\cong} \notag  \\
   =& \   ((W,\rho^w,\sigma^w)_{\cong} \circ  (V,\rho^v,\sigma^v)_{\cong}) \circ  (U,\rho^u,\sigma^u)_{\cong} .
  \end{align} 
The triplet $(1_A,1_I,1_I)_{\cong}: A \rightarrow A $ forms the desired identity morphism, i.e., it is a unit for sequential composition since $(1_A,1_I,1_I)_{\cong}   \circ  (U,\rho^u,\sigma^u)_{\cong}  = ( (1_A \otimes 1_X)    \circ U    ,\rho^u \otimes 1_I,\sigma^u \otimes 1_I)_{\cong}   = (U,\rho^u,\sigma^u)_{\cong} $.
The associativity of parallel composition follows almost identically to the sequential case.
Furthermore, the identity on the unit object $(1_I,1_I,1_I)_{\cong}: I \rightarrow I $ is a unit for parallel composition, since  $(1_I,1_I,1_I)_{\cong}  \otimes  (U,\rho^u,\sigma^u)_{\cong}  = (U \otimes 1_I,\rho^u  \otimes 1_I,\sigma^u  \otimes 1_I)_{\cong}  = (U,\rho^u,\sigma^u)_{\cong} $.
Next, we check the interchange law 
\begin{align*}
    &((U,\rho^u,\sigma^u)_{\cong} \otimes (V,\rho^v,\sigma^v)_{\cong})  \circ((M,\rho^m,\sigma^m)_{\cong} \otimes (N,\rho^n,\sigma^n)_{\cong} )\\
    =& \left(\input{interchange_proof_1.tikz}, \input{interchange_proof_2.tikz}, \input{interchange_proof_3.tikz}\right)_{\cong}
\end{align*}
\begin{align}
    =& \left(\input{interchange_proof_4.tikz}, \input{interchange_proof_5.tikz}, \input{interchange_proof_6.tikz}\right)_{\cong} \notag \\
    =& ((U,\rho^u,\sigma^u)_{\cong} \circ (M,\rho^m,\sigma^m)_{\cong})  \otimes ((V,\rho^v,\sigma^v)_{\cong} \circ (N,\rho^n,\sigma^n)_{\cong} )
\end{align}
where the equality between two equivalence classes follows from
\begin{align}
P\left(\input{interchange_proof_7.tikz}, \input{interchange_proof_8.tikz}\right) \notag\\
\overset{\textbf{(I)}}{=} P\left(\input{interchange_proof_9.tikz}, \input{interchange_proof_10.tikz}\right)
\end{align}
where \textbf{(I)} is applied to swap systems $X'$ and $L'$.
For symmetry, define swap as \[  \text{swap}_{A,B} := \left( \input{sub_thm_swap.tikz}, 1_I, 1_I \right)_{\cong},  \]
which satisfies the necessary axioms for a swap map in a symmetric monoidal category.
\end{proof}

Given that a PPT in the trivial case corresponds to an SPPT, we expect that the quotienting procedure $\mathcal{G}$ for PPTs on  $(\mathcal{C} \subseteq \mathcal{C}, \mathcal{C}(I, A), \mathcal{C}(A, I), \{P^A\})$ should yield the same result as the quotienting procedure $\mathcal{Q}$ for SPPTs. 

\begin{theorem}[\bf Equivalence between quotienting procedures]
  The quotienting procedure $\mathcal{G}$ on probabilistic process theories with trivial subtheory $(\mathcal{C} \subseteq \mathcal{C}, \mathcal{C}(I, A), \mathcal{C}(A, I), \{P^A\})$ yields an SMC isomorphic to the one given by applying the quotienting procedure $\mathcal{Q}$ for simplified probabilistic process theories on $(\mathcal{C}, \{P^A\})$ Def.\ref{def:quotientedcategory}. 
  \label{th:subtheory_isomorphic_quotiented}
\end{theorem}
\begin{proof}
   We show that SMCs $\mathcal{Q}(\mathcal{C}) \text{ and } \mathcal{G}(\mathcal{C})$ are isomorphic (with a strict monoidal functor \cite{Lane}). Let us define functors 
  \begin{align}
    &F(A) = A, \notag\\
    &F\left( \left( \input{sub_thm_1.tikz},\input{sub_thm_2.tikz},\input{sub_thm_3.tikz} \right) _{\cong}\right) = \left\{\input{sub_thm_Usigmarho.tikz}\right\}_{qc},\\
  \end{align}
  and
  \begin{align}
    &G(A) = A, \notag\\
    &G(\{f\}_{qc}) = (f, 1_A, 1_A)_{\cong}
  \end{align}  
  where, the curly braces denote the equivalence class of probabilistically equivalent morphisms. We claim that the functors are well defined on morphisms.  Let us have that $ (U_1,\rho_1,\sigma_1), (U_2,\rho_2,\sigma_2) \in (U,\rho,\sigma)_{\cong}$, i.e.,
  \begin{equation} \forall \tau, \mu: \quad   P \left( \input{sub_defn_1a.tikz}, \input{sub_defn_1b.tikz} \right) =  P \left( \input{sub_defn_2a.tikz}, \input{sub_defn_2b.tikz} \right)  \end{equation}
  but this implies that 
  \begin{align*}
    \input{sub_thm_U1sigma1rho1.tikz} \text{ is probabilistically equivalent to } \input{sub_thm_U2sigma2rho2.tikz}.
  \end{align*}
  Therefore, they are in the same equivalence class in $\mathcal{Q}(\mathcal{C})$. Similarly, let us have that $\{f\} = \{g\}$, i.e.,
  \begin{equation}
  P\left(\input{statefcircrhodouble.tikz}, \input{effectsigmadouble.tikz}\right) = 
 P\left(\input{stategcircrhodouble.tikz}, \input{effectsigmadouble.tikz}\right)
  \end{equation}
  but this implies that $(f, 1_I, 1_I)_{\cong} = (g, 1_I, 1_I)_{\cong}.$ Thus, the functors are defined consistently. To show that these are (strict) symmetric monoidal functors we have
  \begin{align}
    F((1_A, 1_I, 1_I)_{\cong}) &= \{1_A\}_{qc} = 1_{A_\mathcal{Q}(\mathcal{C})}, \\
  G(\{1_A\}_{qc}) &= (1_A, 1_I, 1_I)_{\cong},
  \end{align}
  and 
  \begin{align}
    &F\left( \left( \input{morphismVdouble.tikz},\input{sub_thm_2_v_b.tikz},\input{sub_thm_3_v_b.tikz} \right) _{\cong}\circ \left( \input{sub_thm_1.tikz},\input{sub_thm_2_u.tikz},\input{sub_thm_3_u.tikz} \right) _{\cong}\right) \notag \\
    =& F\left(\left( \input{sub_thm_comp_1z.tikz},\input{sub_thm_comp_2z.tikz},\input{sub_thm_comp_3z.tikz} \right) _{\cong}\right)    =\left\{\input{sub_thmUVcomp.tikz}\right\}_{qc} = \left\{\input{sub_thm_Vbetaalpha.tikz}\right\}_{qc} \circ \left\{\input{sub_thm_Usigmarhou.tikz}\right\}_{qc},\\
  &G(\{g\}_{qc} \circ \{f\}_{qc}) = G(\{g \circ f\}_{qc}) = (g \circ f, 1_A, 1_A)_{\cong} = (g, 1_A, 1_A)_{\cong} \circ (f, 1_A, 1_A)_{\cong},
  \end{align}
and 
  \begin{align}
    &F\left( \left( \input{morphismVdoublecd.tikz},\input{sub_thm_2_v_b.tikz},\input{sub_thm_3_v_b.tikz} \right) _{\cong}\otimes \left( \input{sub_thm_1.tikz},\input{sub_thm_2_u.tikz},\input{sub_thm_3_u.tikz} \right) _{\cong}\right) \notag\\
    =& F\left(\left( \input{morphismUtensorV.tikz},\input{sub_thm_comp_2z.tikz},\input{sub_thm_comp_3z.tikz} \right) _{\cong}\right)
    =\left\{\input{sub_thm_UVparcomp.tikz}\right\}_{qc} \notag\\
    &= \left\{\input{sub_thm_Vcd.tikz}\right\}_{qc} \otimes \left\{\input{sub_thm_Usigmarhou.tikz}\right\}_{qc},\\
  \end{align}
  \begin{equation}
    G(\{g\}_{qc} \otimes \{f\}_{qc}) = G(\{g \otimes f\}_{qc}) = (g \otimes f, 1_A, 1_A)_{\cong} = (g, 1_A, 1_A)_{\cong} \otimes (f, 1_A, 1_A)_{\cong}.
  \end{equation}
  Furthermore, swaps are sent to swaps as
  \begin{equation}
    F\left(\left(\input{sub_thm_swap.tikz}, 1_I, 1_I\right)_{\cong}\right) = \left\{\input{sub_thm_swap.tikz}\right\}_{qc}, \text{ and } G\left(\left\{\input{sub_thm_swap.tikz}\right\}_{qc}\right) = \left(\input{sub_thm_swap.tikz}, 1_I, 1_I\right)_{\cong}.
  \end{equation}
Finally, we show that $F \circ G = id$ and $G \circ F = id$ are identity functors, that is the functors are inverses of each other.
\begin{equation}
  F(G(\{f\}_{qc})) = F((f, 1_A, 1_A)_{\cong}) = \{f\}_{qc},
\end{equation}
  \begin{align}
  G\left(F\left( \left( \input{sub_thm_1.tikz},\input{sub_thm_2.tikz},\input{sub_thm_3.tikz} \right) _{\cong}\right)\right) &= G\left(\left\{\input{sub_thm_Usigmarho.tikz}\right\}_{qc}\right) \notag\\
   &= \left(\input{sub_thm_Usigmarho.tikz}, 1_I, 1_I\right)_{\cong} = \left( \input{sub_thm_1.tikz},\input{sub_thm_2.tikz},\input{sub_thm_3.tikz} \right) _{\cong},\\
\end{align}
where the last equality follows from 
\begin{align}
   P \left( \input{P_comp1.tikz}, \input{P_comp2.tikz} \right) \overset{\textbf{(I)}}{=}  P \left( \input{P_comp3.tikz}, \input{P_comp4.tikz} \right).
\end{align}
\end{proof}


Now, we can declare the probability functions and demonstrate that the quotiented theory is in fact an SPPT. 

\begin{definition}[\bf Quotiented probabilistic process theory]
Given a probabilistic process theory $(\mathcal{D} \subseteq \mathcal{C}, \{\mathcal{S}^A\}, \{\mathcal{E}^A\}, \{P^A\}),$ we define its associated quotiented theory $\mathcal{G}(\mathcal{D} \subseteq \mathcal{C}, \{\mathcal{S}^A\}, \{\mathcal{E}^A\}, \{P^A\}) = (\mathcal{G}(\mathcal{D}) \subseteq \mathcal{G}(\mathcal{D}), \{\mathcal{S}_{\mathcal{G}(\mathcal{D})}^A\}, \{\mathcal{E}_{\mathcal{G}(\mathcal{D})}^A\}, \{P_{\mathcal{G}(\mathcal{D})}^A\})$ as follows:
  \begin{itemize}   
    \item For each object $A$ the set of physical states is  
 \[\mathcal{S}_{\mathcal{G}(\mathcal{D})}^A := \left( \input{sub_thm_1_state.tikz},\input{sub_thm_2.tikz},\input{sub_thm_3.tikz} \right) _{\cong} =  \left( \input{sub_thm_1_state2.tikz},1_I,1_I \right) _{\cong}= \mathcal{G}(\mathcal{D})(I, A)   \] 
 \item For each object $A$ the set of physical effects is  \[ \mathcal{E}_{\mathcal{G}(\mathcal{D})}^A :=   \left( \input{sub_thm_1_effect.tikz},\input{sub_thm_2.tikz},\input{sub_thm_3.tikz} \right) _{\cong} =  \left( \input{sub_thm_1_effect2.tikz},1_I,1_I \right) _{\cong} = \mathcal{G}(\mathcal{D})(A, I) \]
     \item The probability function $P_{\mathcal{G}(\mathcal{D})}^A : \mathcal{S}_{\mathcal{G}(\mathcal{D})}^A \times \mathcal{E}_{\mathcal{G}(\mathcal{D})}^A \rightarrow \mathbb{R}_{\geq 0}$ is given by 
\[  P_{\mathcal{G}(\mathcal{D})}  \left(  \left( \input{sub_thm_1_state.tikz},\input{sub_thm_2_u.tikz},\input{sub_thm_3_u.tikz} \right)_{\cong}  ,  \left( \input{sub_thm_1_effect_v.tikz},\input{sub_thm_2_v.tikz},\input{sub_thm_3_v.tikz} \right) _{\cong}  \right) = P \left(   \input{new_prob_1.tikz}   ,    \input{sub_thm_comp_3.tikz}    \right)  \]
  \end{itemize}
\end{definition}

\begin{theorem}[Quotienting returns an SPPT]
The above theory $\mathcal{G}(\mathcal{D}) \subseteq \mathcal{G}(\mathcal{D})$ with physical states and effects $\mathcal{S}_{\mathcal{G}(\mathcal{D})}^A, \mathcal{E}_{\mathcal{G}(\mathcal{D})}^A$ and probability $P_{\mathcal{G}(\mathcal{D})}$ satisfies the defining axioms for probabilistic process theories Def. \ref{def:physicaltheory} 
\end{theorem}
\begin{proof}
  The first four conditions for physical states and effects follow directly from the definition of states and effects as hom-sets in $\mathcal{G}(\mathcal{D})$. 
  \begin{itemize}
  \item $(s_A, s_B) \in \mathcal{S}_{\mathcal{G}(\mathcal{D})}^A \times \mathcal{S}_{\mathcal{G}(\mathcal{D})}^B \implies s_A \otimes s_B \in \mathcal{S}_{\mathcal{G}(\mathcal{D})}^{A \otimes B}.$
  \item $(e_A, e_B) \in \mathcal{E}_{\mathcal{G}(\mathcal{D})}^A \times \mathcal{E}_{\mathcal{G}(\mathcal{D})}^B \implies e_A \otimes e_B \in \mathcal{E}_{\mathcal{G}(\mathcal{D})}^{A \otimes B}.$
  \item $\forall f \in \mathcal{G}(\mathcal{D})(A,B),\rho \in \mathcal{S}_{\mathcal{G}(\mathcal{D})}^A, \text{ and } \sigma \in \mathcal{E}_{\mathcal{G}(\mathcal{D})}^A, ~~ f \circ \rho \in \mathcal{S}_{\mathcal{G}(\mathcal{D})}^B, \text{ and } \sigma \circ f \in \mathcal{E}_{\mathcal{G}(\mathcal{D})}^A.$
  \item $(1_I, 1_I, 1_I)_{\cong} \in \mathcal{S}_{\mathcal{G}(\mathcal{D})}^I$ and $(1_I, 1_I, 1_I)_{\cong} \in \mathcal{E}_{\mathcal{G}(\mathcal{D})}^I$.
  \end{itemize}
We now check the three constraints for probability functions: \begin{align*}
   \textbf{(I)}: \hspace{1cm} &P_{\mathcal{G}(\mathcal{D})}  \left(   \left( \input{sub_thm_1_w.tikz},\input{sub_thm_2_w.tikz},\input{sub_thm_3_w.tikz} \right)_{\cong} \circ \left( \input{sub_thm_1_state.tikz},\input{sub_thm_2_u.tikz},\input{sub_thm_3_u.tikz} \right)_{\cong}  ,  \left( \input{sub_thm_1_effect_v_b.tikz},\input{sub_thm_2_v_b.tikz},\input{sub_thm_3_v_b.tikz} \right) _{\cong}  \right)  \\
    =&P_{\mathcal{G}(\mathcal{D})}  \left(   \left( \input{sub_thm_1_statew.tikz},\input{sub_thm_2uw.tikz},\input{sub_thm_3uw.tikz} \right)_{\cong} ,  \left( \input{sub_thm_1_effect_v_b.tikz},\input{sub_thm_2_v_b.tikz},\input{sub_thm_3_v_b.tikz} \right) _{\cong}  \right)  \\
\end{align*}
\begin{align}
    =&P  \left(  \input{sub_thm_1_statewv.tikz}, \input{sub_thm_3uwv.tikz}  \right) \notag \\
    =&P_{\mathcal{G}(\mathcal{D})}  \left(  \left( \input{sub_thm_1_state.tikz},\input{sub_thm_2_u.tikz},\input{sub_thm_3_u.tikz} \right)_{\cong} ,  \left( \input{sub_thm_1_effectwv.tikz},\input{sub_thm_2wv.tikz},\input{sub_thm_3wv.tikz} \right) _{\cong}  \right) \notag \\
   = &P_{\mathcal{G}(\mathcal{D})}  \left(   \left( \input{sub_thm_1_state.tikz},\input{sub_thm_2_u.tikz},\input{sub_thm_3_u.tikz} \right)_{\cong}  ,  \left( \input{sub_thm_1_effect_v_b.tikz},\input{sub_thm_2_v_b.tikz},\input{sub_thm_3_v_b.tikz} \right) _{\cong}   \circ \left( \input{sub_thm_1_w.tikz},\input{sub_thm_2_w.tikz},\input{sub_thm_3_w.tikz} \right)_{\cong}  \right).  
\end{align}
\begin{align*}
   \textbf{(II)}: \hspace{0.2cm} &P_{\mathcal{G}(\mathcal{D})}  \left(  \left( \input{sub_thm_1_state.tikz},\input{sub_thm_2_u.tikz},\input{sub_thm_3_u.tikz} \right)_{\cong} \otimes    \left( \input{sub_thm_1_staten.tikz},\input{sub_thm_2_n.tikz},\input{sub_thm_3_n.tikz} \right)_{\cong}, \right.\\
   & \hspace{4cm} \left.  \left( \input{sub_thm_1_effectvaz.tikz},\input{sub_thm_2_v_b.tikz},\input{sub_thm_3_v_b.tikz} \right) _{\cong} \otimes  \left( \input{sub_thm_1_effectmbw.tikz},\input{sub_thm_2_m.tikz},\input{sub_thm_3_m.tikz} \right) _{\cong}  \right)  \\
    =&P_{\mathcal{G}(\mathcal{D})}  \left(   \left( \input{sub_thm_1_stateun.tikz},\input{sub_thm_2_un.tikz},\input{sub_thm_3_un.tikz} \right)_{\cong} , \left( \input{sub_thm_1_effectvm.tikz},\input{sub_thm_2_vm.tikz},\input{sub_thm_3_vm.tikz} \right)_{\cong} \right) 
\end{align*}
\begin{align}
    =&P  \left(  \input{sub_thm_1_compunvm.tikz}, \input{sub_thm_3_unvm.tikz}  \right) \notag\\
     \overset{\textbf{(I)}}{=}& P  \left(  \input{sub_thm_1_compunvm_split.tikz}, \input{sub_thm_3_unvm_split.tikz}  \right) ~~\text{ where } \textbf{(I)} \text{ swaps } Y', Z' \notag\\
     \overset{\textbf{(II)}}{=}& P  \left(  \input{sub_thm_1_compuv.tikz}, \input{sub_thm_3_uv.tikz}  \right) \cdot P  \left(  \input{sub_thm_1_compnm.tikz}, \input{sub_thm_3_nm.tikz}  \right) \notag\\
   =&P_{\mathcal{G}(\mathcal{D})}  \left(  \left( \input{sub_thm_1_state.tikz},\input{sub_thm_2_u.tikz},\input{sub_thm_3_u.tikz} \right)_{\cong} ,\left( \input{sub_thm_1_effectvaz.tikz},\input{sub_thm_2_v_b.tikz},\input{sub_thm_3_v_b.tikz} \right) _{\cong}   \right) \notag\\
   & \hspace{4cm} \cdot P_{\mathcal{G}(\mathcal{D})}\left( \left( \input{sub_thm_1_staten.tikz},\input{sub_thm_2_n.tikz},\input{sub_thm_3_n.tikz} \right)_{\cong}, \left( \input{sub_thm_1_effectmbw.tikz},\input{sub_thm_2_m.tikz},\input{sub_thm_3_m.tikz} \right) _{\cong}  \right).  
\end{align}
    \textbf{(III)}: There exists a state-effect pair $(\rho, \sigma)$ such that $P(\rho, \sigma) \neq 0.$ If we define  $\rho_{gd} = \left( \rho \otimes 1_I , 1_I, 1_I \right)_{\cong}$ and $\sigma_{gd} = \left( \sigma \otimes 1_I, 1_I, 1_I \right)_{\cong}$ then $P_{\mathcal{G}(\mathcal{D})}(\rho_{gd}, \sigma_{gd}) = P(\rho, \sigma) \neq 0.$ Similarly, there exists a state-effect pair $(\rho', \sigma')$  such that $P(\rho', \sigma') \neq 1.$ If we define  $\rho'_{gd} = \left( \rho' \otimes 1_I , 1_I, 1_I \right)_{\cong}$  and $\sigma'_{gd} = \left( \sigma' \otimes 1_I , 1_I, 1_I \right)_{\cong}$ then, $P_{\mathcal{G}(\mathcal{D})}(\rho'_{gd}, \sigma'_{gd}) = P(\rho', \sigma') \neq 1$. 
\end{proof}



We see that the states and effects are simply defined as the hom-sets of morphisms from and to the monoidal unit in $\mathcal{G}(\mathcal{D})$. $\mathcal{G}(\mathcal{D})$ being the physical subcategory of itself, the procedure then, couples back the role of physical states, effects and morphisms. 
In other words, theories such as textbook quantum mechanics which do not strictly form categories, can be lifted to ones that do. 
As a result, the outcome of quotienting a general probabilistic process theory is a simplified probabilistic process theory. 
All results which hold for simplified probabilistic process theories hence, hold for $\mathcal{G}(\mathcal{D})$, such as the existence of a monoid homomorphism $\lambda$. 

\begin{theorem}[\bf Characterisation of quotiented normalised quantum theory] The quotiented theory for the case of normalised quantum theory $(\mathcal{U} \subseteq \mathbf{FHilb}, \{\ket{\psi}\}, \{\bra{\phi}\}, |\langle - | - \rangle|^2)$ yields an SPPT isomorphic to rank-1 CPTNI maps, i.e., Kraus operators $(\mathcal{K}, \Tr[-\circ-])$ where $\mathcal{U}$ is the SMC unitaries, and $\mathcal{K}$ is the SMC of rank-1 CPTNI maps. 
  \label{th:quotiented_quantum_is_standard_quantum}
\end{theorem}
\begin{proof}
The morphisms of $\mathcal{G}(\mathcal{U})$ are of the form \begin{equation}
  \left( \input{sub_thm_1.tikz},\input{sub_thm_2.tikz},\input{sub_thm_3.tikz} \right) _{\cong}.
\end{equation}
We declare a functor $F: \mathcal{G}(\mathcal{U}) \rightarrow \mathcal{K}$ as
\begin{equation}
  \left(\left( \input{sub_thm_1.tikz},\input{sub_thm_2.tikz},\input{sub_thm_3.tikz} \right) _{\cong}\right) = \input{sub_thm_Usigmarho.tikz}.
\end{equation}

Conversely by Stinespring dilation theorem \cite{Nielsen2010Dec, stinespring1955positive} any rank-1 CPTNI map can be dilated to a unitary map acting on a larger Hilbert space with an initialised ancilla state and a final effect $$\input{morphismf.tikz} = \input{sub_thm_Usigmarho.tikz}.$$ Thus, we can define a functor $G: \mathcal{K} \rightarrow \mathcal{G}(\mathcal{U})$ as
\begin{equation}
  G\left(\input{morphismf.tikz}\right) = \left( \input{sub_thm_1.tikz},\input{sub_thm_2.tikz},\input{sub_thm_3.tikz} \right) _{\cong}.
\end{equation}
 That $G$ is well-defined, i.e., the outcome of $G$ does not depend on the chosen dilation unitary is a consequence of probabilistic equivalence  Def. \ref{def:subtheory}.  These being symmetric monoidal functors can be proven identically to Th. \ref{th:subtheory_isomorphic_quotiented}. To show that these are inverses of each other we have \begin{equation}F\left(G\left(\input{morphismf.tikz}\right)\right) = F\left(\left( \input{sub_thm_1.tikz},\input{sub_thm_2.tikz},\input{sub_thm_3.tikz} \right) _{\cong}\right) = \input{sub_thm_Usigmarho.tikz} = \input{morphismf.tikz}
 \end{equation}
  and 
  \begin{equation}
    G\left(F\left(\left( \input{sub_thm_1.tikz},\input{sub_thm_2.tikz},\input{sub_thm_3.tikz} \right) _{\cong}\right)\right) = G\left(\input{sub_thm_Usigmarho.tikz}\right) = \left( \input{sub_thm_1.tikz},\input{sub_thm_2.tikz},\input{sub_thm_3.tikz} \right) _{\cong}.
  \end{equation}
   Thus, $F \circ G = id$ and $G \circ F = id$ are identity functors, and the two categories are isomorphic. For probability functions, we have
\begin{eqnarray} &&P_{\mathcal{G}(\mathcal{U})}  \left(  \left( \input{sub_thm_1_state.tikz},\input{sub_thm_2_u.tikz},\input{sub_thm_3_u.tikz} \right)_{\cong}  ,  \left( \input{sub_thm_1_effect_v.tikz},\input{sub_thm_2_v.tikz},\input{sub_thm_3_v.tikz} \right) _{\cong}  \right) = P \left(   \input{new_prob_1.tikz}   ,    \input{sub_thm_comp_3.tikz}    \right) \notag\\
&&= \left|\input{new_prob_2.tikz}\right|^2 = \Tr[\rho \sigma], \end{eqnarray}
where 
\begin{equation}
  \rho = \input{new_prob_3.tikz} \text{ and } \sigma = \input{new_prob_4.tikz}.
\end{equation}
\end{proof}


Let us now outline the explicit consequences of the fact that the outcome of the quotienting procedure $\mathcal{G}$ is simplified probabilistic process theory.
\begin{corollary}
 For the (simplified) probabilistic process theory $(\mathcal{G}(\mathcal{D}), {P}_{\mathcal{G}(\mathcal{D})}^A)$, the probability function $P_{\mathcal{G}(\mathcal{D})}$ can be computed by state-effect composition followed by a monoid homomorphism $\lambda_{\mathcal{G}(\mathcal{D})} : \mathcal{D}(I,I) \rightarrow \mathbb{R}_{\geq 0}$ by \[    \lambda_{\mathcal{G}(\mathcal{D})}  \left(   \left( \input{sub_thm_scalar.tikz},\input{sub_thm_2.tikz},\input{sub_thm_3.tikz} \right)_{\cong}   \right) = P_{\mathcal{G}(\mathcal{D})}\left((1_I, 1_I, 1_I)_{\cong}, \left( \input{sub_thm_scalar.tikz},\input{sub_thm_2.tikz},\input{sub_thm_3.tikz} \right)_{\cong}\right) \overset{\text{Def.}}{=} P\left(\input{sub_thm_scalar_2.tikz},\input{sub_thm_3.tikz}\right) \]
  \label{lem:lambda_is_bornrule}
\end{corollary}
\begin{proof}
Immediate by the results of the previous section.
\end{proof}

To be explicit the probability for a state-effect pair 
  \[\left(  \left( \input{sub_thm_1_state.tikz},\input{sub_thm_2_u.tikz},\input{sub_thm_3_u.tikz} \right)_{\cong},  \left( \input{sub_thm_1_effectvaz.tikz},\input{sub_thm_2_v_b.tikz},\input{sub_thm_3_v_b.tikz} \right) _{\cong}  \right)\]
  is given by
\begin{align}
  & P_{\mathcal{G}(\mathcal{D})}  \left(  \left( \input{sub_thm_1_state.tikz},\input{sub_thm_2_u.tikz},\input{sub_thm_3_u.tikz} \right)_{\cong}  ,  \left( \input{sub_thm_1_effectvaz.tikz},\input{sub_thm_2_v_b.tikz},\input{sub_thm_3_v_b.tikz} \right) _{\cong}  \right) \notag\\
  &= P_{\mathcal{G}(\mathcal{D})}  \left( (1_I,1_I,1_I)_{\cong}  ,  \left( \input{sub_thm_1_effectvaz.tikz},\input{sub_thm_2_v_b.tikz},\input{sub_thm_3_v_b.tikz} \right) _{\cong} \circ \left( \input{sub_thm_1_state.tikz},\input{sub_thm_2_u.tikz},\input{sub_thm_3_u.tikz} \right)_{\cong}   \right) \notag\\
  &=  \lambda_{\mathcal{G}(\mathcal{D})}  \left(    \left( \input{sub_thm_1_effectvaz.tikz},\input{sub_thm_2_v_b.tikz},\input{sub_thm_3_v_b.tikz} \right) _{\cong} \circ \left( \input{sub_thm_1_state.tikz},\input{sub_thm_2_u.tikz},\input{sub_thm_3_u.tikz} \right)_{\cong}   \right) \notag \\
  &= \lambda_{\mathcal{G}(\mathcal{D})}\left(\left(\input{scalar_compuv.tikz}, \input{sub_thm_2_uv.tikz}, \input{sub_thm_3_uv.tikz}\right)_{\cong}\right).
  \end{align}




To show that the theory $(\mathcal{G}(\mathcal{D}), {P}_{\mathcal{G}(\mathcal{D})}^A)$ has a generalised Born rule, we only need to confirm that the SPPT at hand is itself the result of a quotienting procedure $\mathcal{Q}$ on an SPPT.
\begin{theorem}[\bf Stability under quotienting]
  The quotienting procedure $\mathcal{Q}$ on the theory $(\mathcal{G}(\mathcal{D}), {P}_{\mathcal{G}(\mathcal{D})}^A)$ yields  $(\mathcal{G}(\mathcal{D}), {P}_{\mathcal{G}(\mathcal{D})}^A)$.  
\end{theorem}
\begin{proof} 
  We simply need to show that if two morphisms in $\mathcal{G}(\mathcal{D})$ are probabilistically equivalent then they are the same. Let us have that two morphisms $(U_1, \rho_1, \sigma_1)_{\cong}$ and $(U_2, \rho_2, \sigma_2)_{\cong}$ are probabilistically equivalent in $\mathcal{G}(\mathcal{D})$. This means that for any state-effect pair $(\tau, \rho^\tau, \sigma^\tau)_{\cong}, (\mu, \rho^{\mu}, \sigma^{\mu})_{\cong}$ we should have 
  \begin{align}
    &P_{\mathcal{G}(\mathcal{D})}\left( [(U_1, \rho_1, \sigma_1)_{\cong} \otimes (1_B, 1_I, 1_I )_{\cong}] \circ (\tau, \rho^\tau, \sigma^\tau)_{\cong}, (\mu, \rho^{\mu}, \sigma^{\mu})_{\cong}  \right) \notag\\
    =& P_{\mathcal{G}(\mathcal{D})}\left( [(U_2, \rho_2, \sigma_2)_{\cong} \otimes (1_B, 1_I, 1_I )_{\cong}] \circ (\tau, \rho^\tau, \sigma^\tau)_{\cong}, (\mu, \rho^{\mu}, \sigma^{\mu})_{\cong} \right), 
  \end{align}
  that is
  \begin{align}
    &P_{\mathcal{G}(\mathcal{D})}\left( \left(\input{sub_thm_quo_1.tikz}, \input{sub_thm_quo_2.tikz}, \input{sub_thm_quo_3.tikz} \right)_{\cong}, \left(\input{sub_thm_quo_4.tikz}, \input{sub_thm_quo_5.tikz}, \input{sub_thm_quo_6.tikz} \right)_{\cong} \right) \notag \\
    =& P_{\mathcal{G}(\mathcal{D})}\left( \left(\input{sub_thm_quo_7.tikz}, \input{sub_thm_quo_8.tikz}, \input{sub_thm_quo_9.tikz} \right)_{\cong}, \left(\input{sub_thm_quo_4.tikz}, \input{sub_thm_quo_5.tikz}, \input{sub_thm_quo_6.tikz} \right)_{\cong} \right).
  \end{align}
  By definition of $P_{\mathcal{G}(\mathcal{D})}$ this requirement is equivalent to 
    \begin{align*}
    &P\left(\input{sub_thm_quo_10.tikz}, \input{sub_thm_quo_11.tikz} \right) 
    = P\left(\input{sub_thm_quo_12.tikz}, \input{sub_thm_quo_13.tikz} \right)\\
    \overset{\mathbf{(I)}}{\iff}&P\left(\input{sub_thm_quo_14.tikz}, \input{sub_thm_quo_15.tikz} \right)\\ 
    \end{align*}
    \begin{align}
    =& P\left(\input{sub_thm_quo_16.tikz}, \input{sub_thm_quo_17.tikz} \right).
  \end{align}
  Setting $Y, Y' = I$ and $Z, Z' = I$ with $\rho^\mu, \rho^\tau = 1_I$ and $\sigma^\mu, \sigma^\tau = 1_I$ in the above equation, we have precisely the condition that $(U_1, \rho_1, \sigma_1)\text{ and } (U_2, \rho_2, \sigma_2)$ are probabilistically equivalent and thus $(U_1, \rho_1, \sigma_1)_{\cong}$ and $(U_2, \rho_2, \sigma_2)_{\cong}$ are the same morphisms in $\mathcal{G}(\mathcal{D})$. Thus, the quotienting procedure $\mathcal{Q}$ on $\mathcal{G}(\mathcal{D})$ yields $\mathcal{G}(\mathcal{D})$ itself. Furthermore, since the probability of a state-effect pair in a quotiented theory is just the probability of the particular instances of state-effect pair in their equivalence classes Eq.\eqref{eq:probabilityfunctionQC}, the probability function is also unchanged by the quotienting procedure. Thus, $\mathcal{Q}(\mathcal{G}(\mathcal{D}), {P}_{\mathcal{G}(\mathcal{D})}^A) = (\mathcal{G}(\mathcal{D}), {P}_{\mathcal{G}(\mathcal{D})}^A)$.
\end{proof}

Since, $(\mathcal{G}(\mathcal{D}), {P}_{\mathcal{G}(\mathcal{D})}^A)$ is stable under quotienting $\mathcal{Q}$, all the results from section \ref{sec:probphysicaltheories} are also immediately valid for $(\mathcal{G}(\mathcal{D}), {P}_{\mathcal{G}(\mathcal{D})}^A).$ In particular $\lambda_{\mathcal{G}(\mathcal{D})}$ elevates to a monoid isomorphism along with its inverse $\theta_{\mathcal{G}(\mathcal{D})} = \lambda^{-1}_{\mathcal{G}(\mathcal{D})}$. Consequently, we have the following:
\begin{theorem}[Generalised Born rule for probabilistic process theories]
The quotient $\mathcal{G}$ of a probabilistic process theory satisfies the generalised Born rule, in other words, $\lambda_{\mathcal{G}(\mathcal{D})} : \mathcal{D}(I,I) \rightarrow \mathbb{R}_{\geq 0}$ is injective.
\end{theorem}
\begin{proof}
Follows immediately by the results of Thm. \ref{thm:generalisedBornQC}.
\end{proof}

In summary, the quotienting procedure allows us to systematically construct a simplified probabilistic process theory which satisfies the generalised Born rule from any probabilistic process theory. 
In other words, the operational content of a theory i.e., its observable statistics, can always be represented compositionally. 

  \section{Probabilities in Summed Probabilistic Process Theories}
  \label{sec:summedtheories}

With the aim of strengthening the correspondence between scalars and probabilities and furthermore recovering the completely positive maps from textbook quantum theory, we now consider freely adding noise to probabilistic process theories. We do this by considering weighted sets of morphisms from the process theory $\mathcal{C}$ of an SPPT, thereby building a new process theory $\mathcal{S}(\mathcal{C})$. 

\begin{definition}[\bf Summed Category $\mathcal{S}(\mathcal{C})$]
For the SMC $\mathcal{C}$ of any simplified probabilistic process theory we can construct a new category $\mathcal{S}(\mathcal{C})$ with the same objects as $\mathcal{C}$. Its morphisms are finite sets of ordered pairs with morphisms from $\mathcal{C}$, and non-zero weights from $\mathbb{R}_{> 0}$, i.e., 
\[ \mathcal{S}(\mathcal{C})(A,B) := \left\{ (\mathfrak{f}_i, w_i)_{i \in J} \mid \mathfrak{f}_i \in \mathcal{C}(A,B), w_i \in \mathbb{R}_{>0} \right\}, \]
where $J$ is a finite index set. 
\end{definition}

  \begin{definition}[\bf Sequential composition in $\mathcal{S}(\mathcal{C})$] Given two morphisms $s_1 = \{ (\mathfrak{f}_i, w_i) \}_{sc} \in \mathcal{S}(\mathcal{C})(B,C)$ and $s_2 = \{ (\mathfrak{g}_j, v_j) \}_{sc} \in \mathcal{S}(\mathcal{C})(A,B)$ in $\mathcal{S}(\mathcal{C})$. We define $s_1 \circ_{sc} s_2 = \{ (\mathfrak{f}_i \circ \mathfrak{g}_j, w_i v_j) ~| ~\forall ~ i,j\}_{sc}$.
  \end{definition}

  \begin{definition}[\bf Parallel composition in $\mathcal{S}(\mathcal{C})$] Given two morphisms $s_1 = \{ (\mathfrak{f}_i, w_i) \}_{sc} \in \mathcal{S}(\mathcal{C})(A,B)$ and $s_2 = \{ (\mathfrak{g}_j, v_j) \}_{sc} \in \mathcal{S}(\mathcal{C})(A,B)$ in $\mathcal{S}(\mathcal{C})$. We define $s_1 \otimes_{sc} s_2 = \{ (\mathfrak{f}_i \otimes \mathfrak{g}_j, w_i v_j) ~| ~\forall ~ i,j\}_{sc}$.
  \end{definition}

  \begin{definition}
  We define identity morphism $1_{A_{sc}}$ for an object $A$ as the set of pair $1_{A_{sc}} = \{(1_{A}, 1)\}_{sc}$.
  \end{definition}

  \begin{theorem}
    $\mathcal{S}(\mathcal{C})$ is a symmetric monoidal category.
  \label{th:SCisSMC}
  \end{theorem}
  \begin{proof}
    Here we provide the proof for one of the identities, the rest follow similarly. Given morphisms $\mathfrak{g}_1 \in \mathcal{S}(\mathcal{C})(B,C), \mathfrak{g}_2 \in \mathcal{S}(\mathcal{C})(E, F), \mathfrak{f}_1 \in \mathcal{S}(\mathcal{C})(A, B),$ and $\mathfrak{f}_2 \in \mathcal{S}(\mathcal{C})(D, E)$, we prove the interchange law:
    \begin{eqnarray}
      &&(\mathfrak{g}_1 \otimes_{sc} \mathfrak{g}_2) \circ_{sc} (\mathfrak{f}_1 \otimes_{sc} \mathfrak{f}_2) \notag\\
      &=& \{((g_1 \otimes g_2) \circ (f_1 \otimes f_2), w_1 w_2 v_1 v_2) | (g_1, w_1) \in \mathfrak{g}_1, (g_2, w_2) \in \mathfrak{g}_2, (f_1, v_1) \in \mathfrak{f}_1, (f_2, v_2) \in \mathfrak{f}_2 \}_{sc} \notag\\
      &=& \{((g_1 \circ f_1) \otimes (g_2 \circ f_2), w_1 w_2 v_1 v_2) | (g_1, w_1) \in \mathfrak{g}_1, (g_2, w_2) \in \mathfrak{g}_2, (f_1, v_1) \in \mathfrak{f}_1, (f_2, v_2) \in \mathfrak{f}_2 \}_{sc} \notag\\
      &=& (\mathfrak{g}_1 \circ_{sc} \mathfrak{f}_1) \otimes_{sc} (\mathfrak{g}_2 \circ_{sc} \mathfrak{f}_2)
    \end{eqnarray}
    
    In general, because of associativity of cartesian product, associativity of product in $\mathbb{R}_{> 0}$ and the fact that $\mathcal{C}$ is a symmetric monoidal category, we have that $\mathcal{S}(\mathcal{C})$ is a symmetric monoidal category.
  \end{proof}

Considering our main examples of interest, morphisms in summed category of $(\mathbf{FHilb}, |\langle -|- \rangle|^2)$ are simply the weighted sets of complex linear maps between finite-dimensional Hilbert spaces. Similarly, morphisms in the summed category of quotiented normalised quantum theory $(\mathcal{K}, \Tr[-\circ-])$ are weighted set of Kraus operators.

The definition of morphisms as finite weighted sets naturally allows for the inclusion of the empty set, which we identify as the zero morphism \cite{Pareigis1970Jan}.
  \begin{definition}[\bf Zero morphism in $\mathcal{S}(\mathcal{C})$] zero morphism $0_{sc}$ in $\mathcal{S}(\mathcal{C})$ is defined as the empty set, denoted $0_{sc} = \{\}$. A unique zero morphism exists for each pair of objects, $0_{sc} = \{\} \in \mathcal{S}(\mathcal{C})(A,B)$ for all $A,B$\footnote{We use the same notation for all zero morphisms}.
  \end{definition}

  \begin{lemma}[\bf Properties of the zero morphism] The zero morphism satisfies the expected properties:
    \begin{itemize}
    \item $0_{sc} \otimes_{sc} f = 0_{sc},  f \otimes_{sc} 0_{sc}  = 0_{sc}~~\forall ~ f \in \mathcal{S}(\mathcal{C})(A,B)$.
    \item $0_{sc} \circ_{sc} f = 0_{sc}, f \circ_{sc} 0_{sc}  = 0_{sc} ~~ \forall ~f \in \mathcal{S}(\mathcal{C})(A,B)$.
  \end{itemize}
  \label{lemma:zeroMorproperties}
  \end{lemma}

  \begin{proof}
    In any of the compositions, there are precisely zero elements to be selected from the zero morphism. Hence there are zero pairs inside the composition, creating an empty set, i.e., a zero morphism. 
  \end{proof}

A zero morphism in an SPPT is expected to have zero probability. We show that this is indeed the case due to axioms \textbf{I}, \textbf{II}, and \textbf{III}.

  \begin{theorem}
    Given a probabilistic process theory $(\mathcal{C}, \{P^A\})$ with a unique zero morphism $\mathbf{0}$ for each pair of objects $(A, B)$ satisfying $\mathbf{0} \otimes f = \mathbf{0},  f \otimes \mathbf{0}  = \mathbf{0}~~\forall ~ f \in \mathcal{C}(A,B)$. The probability $P(1_I, \mathbf{0})$ evaluates to zero.
    \label{th:SCzeromorphism}
  \end{theorem}
  \begin{proof}
    By axiom \textbf{(III)}, let us assume that there exists a state-effect pair $(\rho, \sigma)$ such that $P(\rho, \sigma) \neq 1$. We consider
    \begin{eqnarray}
      P(1_I, \mathbf{0}) 
      &\overset{1_I \otimes 1_I = 1_I, \, \mathbf{0} \otimes \gamma = \mathbf{0}}{=}& P(1_{I} \otimes 1_I, \mathbf{0} \otimes (\sigma \circ \rho)) \notag\\
      &\overset{\bf (I)}{=}& P(1_{I} \otimes \rho, \mathbf{0} \otimes \sigma)  \notag\\
      &\overset{\bf (II)}{=}&  P(1_{I}, \mathbf{0}) P(\rho, \sigma) \notag\\
      &\implies& [1-P(\rho, \sigma)] P(1_{I}, \mathbf{0}) = 0 
      \end{eqnarray}
  Since $P(\rho, \sigma) \neq 1$, we have $[1-P(\rho, \sigma)] P(1_{I}, \mathbf{0}) = 0  \implies  P(1_{I}, \mathbf{0}) = 0$.
  \end{proof}

With the summed category $\mathcal{S}(\mathcal{C})$ defined, we now proceed to define the probability function for it. Much like what we did for the quotiented theory, we define the probability function for the summed category $\mathcal{S}(\mathcal{C})$ by lifting it from the probability function of an SPPT.

  \begin{definition}[\bf Probabilities in $\mathcal{S}(\mathcal{C})$] Let $\rho_{sc}$ and $\sigma_{sc}$ be the states and effects in the summed category $\mathcal{S}(\mathcal{C})$. We define the family of probability functions $P^A_{\mathcal{S}(\mathcal{C})}:\mathcal{S}(\mathcal{C})(I, A) \times \mathcal{S}(\mathcal{C})(A, I) \rightarrow \mathbb{R}_{\geq 0}$ as follows: 
    \[
    P^A_{\mathcal{S}(\mathcal{C})}\left(\input{staterhoSC.tikz}, \input{effectsigmaSC.tikz}\right) = \sum_{i,j} w_i v_j P_{\mathcal{C}}\left(\input{staterhoQCi.tikz}, \input{effectsigmaQCj.tikz}\right) \quad \text{ for }  \left(\input{staterhoQCi.tikz}, w_i\right) \in \input{staterhoSC.tikz} \text{ and } \left(\input{effectsigmaQCj.tikz}, v_j\right) \in \input{effectsigmaSC.tikz},  
    \]
    \[P^A_{\mathcal{S}(\mathcal{C})}\left(\rho, 0_{sc}\right) = P^A_{\mathcal{S}(\mathcal{C})}\left(0_{sc}, \sigma\right) = 0 ~~\forall \rho, \sigma.\]
  \end{definition}

These weighted sets serve two distinct roles. Intuitively the weighted sets for morphisms in the summed category correspond to imbuing classical noise in an SPPT, where the weights can be thought of as probabilities of the occurrence of those morphisms. Thus, the summed category $\mathcal{S}(\mathcal{C})$ is a noisy version of $\mathcal{C}$. At the same time, weighted sets of scalars in the summed category yield sums of probabilities and therefore an additive structure in the theory.

We now show that the family of probability functions $P^A_{\mathcal{S}(\mathcal{C})}$ satisfy the axioms of a simplified probabilistic process theory.

  \begin{theorem}
    The probability function $P^A_{\mathcal{S}(\mathcal{C})}$ satisfies the axioms \textbf{(I)}, \textbf{(II)}, and \textbf{(III)} in the category $\mathcal{S}(\mathcal{C})$.
  \end{theorem}

  \begin{proof}
    \textbf{(I)}: Let $\rho_{sc} \in \mathcal{S}(\mathcal{C})(I,A)$, $\sigma_{sc} \in \mathcal{S}(\mathcal{C})(B,I)$ and $f_{sc} \in \mathcal{S}(\mathcal{C})(A,B)$ be the states, effects and morphisms in the summed category $\mathcal{S}(\mathcal{C})$ and let $(\rho_i, w_i) \in \rho_{sc}, (\sigma_j, v_j) \in \sigma_{sc}$ and $(f_k, u_k) \in f_{sc}$. Then
    \begin{eqnarray}
    P^B_{\mathcal{S}(\mathcal{C})}\left(\input{statefcircrhodoubleSC.tikz},\input{effectsigmadoublesc.tikz}\right) &=& P^B_{\mathcal{S}(\mathcal{C})}\left(\left\{\left(\input{statefkcircrhoidoubleQC.tikz}, u_k w_i\right)\right\}_{sc}, \left\{\left(\input{effectsigmadoubleqcj.tikz}, v_j\right)\right\}_{sc}\right)\notag\\
    &\overset{\text{Def.}}{=}& \sum_{k,i,j} u_k w_i v_j P^B \left(\input{statefkcircrhoidoubleQC.tikz}, \input{effectsigmadoubleqcj.tikz}\right) \notag\\
      &\overset{Def. \ref{def:physicaltheory}}{=}& \sum_{k,i,j} u_k w_i v_j P^B\left(\input{staterhodoubleqci.tikz}, \input{effectsigmajcircfkdoubleQC.tikz} \right) \notag\\
      &=& P^B_{\mathcal{S}(\mathcal{C})}\left(\input{staterhodoublesc.tikz},\input{effectsigmacircfdoubleSC.tikz}\right)
    \end{eqnarray}
    \textbf{(II)}:  Let $\rho_{sc} \in \mathcal{S}(\mathcal{C})(I,A), \rho'_{sc} \in \mathcal{S}(\mathcal{C})(I,B)$ and  $\sigma_{sc} \in \mathcal{S}(\mathcal{C})(A,I), \sigma'_{sc} \in \mathcal{S}(\mathcal{C})(B,I)$ be states and effects in the summed category $\mathcal{S}(\mathcal{C})$ and let $(\rho_i, w_i) \in \rho_{sc}, (\rho'_j, w'_j) \in \rho'_{sc}, (\sigma_k, u_k) \in \sigma_{sc}, (\sigma'_l, u'_l) \in \sigma'_{sc}$. Then
    \begin{eqnarray}
    P^{A \otimes B}_{\mathcal{S}(\mathcal{C})}\left(\input{staterhoSC.tikz} \input{staterhoprimeSC.tikz},\input{effectsigmaSC.tikz} \input{effectsigmaprimeSC.tikz}\right) &=& P^{A \otimes B}_{\mathcal{S}(\mathcal{C})}\left(\left\{\bigg(\input{staterhoQCi.tikz} \input{staterhoQCprimej.tikz}, w_i w'_j\bigg)\right\}_{sc}, \left\{\bigg(\input{effectsigmaQCk.tikz} \input{effectsigmaQCprimel.tikz}, u_k u'_l\bigg)\right\}_{sc}\right)\notag\\
    &\overset{\text{Def.}}{=}& \sum_{i,j,k,l} w_i w'_j u_k u'_l P^{A \otimes B}_{}\left(\input{staterhoQCi.tikz} \input{staterhoQCprimej.tikz}, \input{effectsigmaQCk.tikz} \input{effectsigmaQCprimel.tikz}\right) \notag\\
    &\overset{Def. \ref{def:physicaltheory}}{=}& \sum_{i,j,k,l} w_i w'_j u_k u'_l P^{A }_{}\left(\input{staterhoQCi.tikz}, \input{effectsigmaQCk.tikz} \right) \cdot P^{ B}_{}\left(\input{staterhoQCprimej.tikz}, \input{effectsigmaQCprimel.tikz}\right)  \notag\\
      &=&  P^{A}_{\mathcal{S}(\mathcal{C})}\left(\input{staterhoSC.tikz},\input{effectsigmaSC.tikz} \right) \cdot P^{B}_{\mathcal{S}(\mathcal{C})}\left(\input{staterhoprimeSC.tikz},\input{effectsigmaprimeSC.tikz} \right)
    \end{eqnarray}
    \textbf{(III)}: Since, there exists a pair $(\rho, \sigma)$ such that $P_{}(\rho, \sigma) \neq 0$, we can choose $\rho_{sc} \in \mathcal{S}(\mathcal{C})(I,A)$ and $\sigma_{sc} \in \mathcal{S}(\mathcal{C})(B,I)$ be  $\{(\rho, 1)\}_{sc}$ and $\{(\sigma, 1)\}_{sc}$ respectively. Then, we have $P_{\mathcal{S}(\mathcal{C})}(\rho_{sc}, \sigma_{sc}) \neq 0.$ Similarly, there exists a pair $(\rho', \sigma')$ such that $P_{}(\rho', \sigma') \neq 1$, we can choose $\rho'_{sc} \in \mathcal{S}(\mathcal{C})(I,A)$ and $\sigma'_{sc} \in \mathcal{S}(\mathcal{C})(B,I)$ be  $\{(\rho', 1)\}_{sc}$ and $\{(\sigma', 1)\}_{sc}$ respectively. So that $P_{\mathcal{S}(\mathcal{C})}(\rho'_{sc}, \sigma'_{sc}) \neq 1.$
  \end{proof}

  With the above theorem, we have shown that the summed category $\mathcal{S}(\mathcal{C})$ along with the family of probability functions $P^A_{\mathcal{S}(\mathcal{C})},$ which we term as a summed probabilistic process theory satisfies the axioms of an SPPT.

  \begin{corollary}
    $\left(\mathcal{S}(\mathcal{C}), \{P^A_{\mathcal{S}(\mathcal{C})}\}\right)$ is an SPPT.
    \label{cor:SCvalidtheory}
  \end{corollary}

  Much like an SPPT, the summed category $\mathcal{S}(\mathcal{C})$ also has a monoid homomorphism from its scalars to probabilities. However this time around the homomorphism is a bit different. The scalars in the summed category $\mathcal{S}(\mathcal{C})$ are weighted sets of scalars from the category $\mathcal{C}$ and the probabilities in the summed category $\mathcal{S}(\mathcal{C})$ are sums of probabilities from $\mathcal{C}$. Thus, the monoid homomorphism from scalars to probabilities in the summed category $\mathcal{S}(\mathcal{C})$ also has a sum structure to preserve. We show that this is indeed the case. It is this sum structure, introduced by classical noise, which helps us further confine the possible monoid homomorphisms from scalars to probabilities.

  \begin{definition}
    We define $\lambda_{\mathcal{S}(\mathcal{C})}:\mathcal{S}(\mathcal{C})(I,I) \rightarrow \mathbb{R}_{\geq 0}$ as follows:
    \[
    \lambda_{\mathcal{S}(\mathcal{C})}(\gamma_{sc}) = P_{\mathcal{S}(\mathcal{C})}(1_I, \gamma_{sc}) = \sum_{i} w_i P(1_I, \gamma_i) \text{ for } (\gamma_i, w_i) \in \gamma_{sc}.
    \]
  \end{definition}

An alternative way to define $\lambda_{\mathcal{S}(\mathcal{C})}$ is by lifting the monoid homomorphism $\lambda$ from the category $\mathcal{C}$ to the summed category $\mathcal{S}(\mathcal{C})$. This is done by defining $\lambda_{\mathcal{S}(\mathcal{C})}$ as $\lambda_{\mathcal{S}(\mathcal{C})}(\{(a_i,\alpha_i)\}) = \sum_{i} \alpha_i \lambda(a_i)$. Both definitions are equivalent by  probabilities in the summed category $\mathcal{S}(\mathcal{C})$.

Since a summed theory is an SPPT, we have the familiar monoid homomorphism from scalars to probabilities.

  \begin{theorem}
    There exists a monoid homomorphism $\lambda_{\mathcal{S}(\mathcal{C})}$ from $\mathcal{S}(\mathcal{C})(I, I)$ to $(\mathbb{R}_{\geq 0}, 1, \times)$.
    \label{th:SCmonoidalhomo}
  \end{theorem}
  \begin{proof}
    In section \ref{sec:probphysicaltheories} we have shown that such a monoid homomorphism exists for all valid simplified probabilistic process theories. Hence, by corollary \ref{cor:SCvalidtheory}, we have that $\lambda_{\mathcal{S}(\mathcal{C})}$ is a monoid homomorphism from $\mathcal{S}(\mathcal{C})(I, I)$ to $(\mathbb{R}_{\geq 0}, 1, \times)$. 
  \end{proof}

  Before presenting a conservation of the sum structure, we begin by formalising it. To this end we define an additive union operation $\cup$ in the summed category $\mathcal{S}(\mathcal{C})$.

  \begin{definition}
    For morphisms $\{(f_i, w_i)\}_{sc} = f_{sc} \in \mathcal{S}(\mathcal{C})(A,B)$ and $\{(g_j, v_j)\}_{sc} = g_{sc} \in \mathcal{S}(\mathcal{C})(C,D)$, we define $f_{sc}  \cup g_{sc}$ as the additive union morphism, with $f_{sc}  \cup g_{sc} = \{(f_i, w_i), (g_j, v_j) ~\forall i, j \}_{sc}$
  \end{definition}

  It is the additive union along with the zero morphism that forms the sum structure in summed probabilistic process theories. The additive union is commutative and associative and the zero morphism serves as the unit. Furthermore, the distributive property holds between the additive union and the monoidal product. We now show that all these properties hold for scalars in the summed category $\mathcal{S}(\mathcal{C})$ and that these form a semiring.

  \begin{theorem}
    $\mathcal{S}(\mathcal{C})(I, I)$ is a semiring. 
    \label{th:SCsemiring}
  \end{theorem}

  \begin{proof}
     We already have $(\mathcal{S}(\mathcal{C})(I, I), 1_{I_{sc}}, \otimes) $ as a monoid. Furthermore $(\mathcal{S}(\mathcal{C})(I, I), 0_{sc}, \cup)$ is a commutative monoid from the following:

    \noindent Associativity: Given morphisms $f_{sc}, g_{sc}, h_{sc} \in \mathcal{S}(\mathcal{C})(I,I)$, we have
    \begin{eqnarray}
      (f_{sc} \cup g_{sc}) \cup h_{sc} &=& \{(f_i, w_i), (g_j, v_j), (h_k, u_k) ~~\forall~ i,j,k\}_{sc} \notag\\
      &=& f_{sc} \cup (g_{sc} \cup h_{sc}).
    \end{eqnarray}
    Identity: The zero morphism (the empty set) serves as the identity element for the additive union $\cup$.

    \noindent Commutativity: Given morphisms $f_{sc}, g_{sc} \in \mathcal{S}(\mathcal{C})(I,I)$, we have
    \begin{eqnarray}
      f_{sc} \cup g_{sc}  &=& \{(f_i, w_i), (g_j, v_j) ~\forall i,j\}_{sc} \notag\\
      &=&  g_{sc} \cup f_{sc}. 
    \end{eqnarray}
    By the properties of zero morphism, lemma \ref{lemma:zeroMorproperties}, we have $f_{sc} \otimes_{sc} 0_{sc} =  0_{sc} \otimes_{sc} f_{sc} = 0_{sc}.$
    Lastly, given morphisms $f_{sc}, g_{sc}, h_{sc} \in \mathcal{S}(\mathcal{C})(I,I)$, we have
    \begin{eqnarray}
      f_{sc} \otimes_{sc} (g_{sc} \cup h_{sc}) &=& \{(f_i \otimes g_j, w_i v_j), (f_k \otimes h_l, w_k u_l)\}_{sc} \notag\\
      &=& (f_{sc} \otimes_{sc} g_{sc}) \cup (f_{sc} \otimes_{sc} h_{sc})
    \end{eqnarray}
    similarly, 
      \begin{eqnarray}
      (g_{sc} \cup h_{sc}) \otimes_{sc} f_{sc}  &=& \{( g_j  \otimes f_i, v_j w_i), ( h_l  \otimes f_k, w_k u_l )\}_{sc} \notag\\
      &=& ( g_{sc} \otimes_{sc} f_{sc}) \cup (h_{sc}  \otimes_{sc} f_{sc} )
    \end{eqnarray}
    Hence, $\mathcal{S}(\mathcal{C})(I, I)$ is indeed a semiring.
  \end{proof}

  With all the constituents at hand, we can prove the proposed semiring homomorphism from scalars to probabilities in the summed theory $(\mathcal{S}(\mathcal{C}),\{P_{\mathcal{S}(\mathcal{C})}^A\})$. The core of the proof is to show that the monoid homomorphism $\lambda_{\mathcal{S}(\mathcal{C})}$ also preserves the additive union. 

  \begin{theorem}[\bf Born rule as semiring homomorphism in summed theories]
    $\lambda_{\mathcal{S}(\mathcal{C})}$ is a semiring homomorphism from $(\mathcal{S}(\mathcal{C})(I, I), 0_{sc}, \cup, 1_{I_{sc}}, \otimes_{\mathcal{S}} )$ to $(Range(P_{\mathcal{S}(\mathcal{C})}) \subseteq \mathbb{R}_{\geq 0}, 0, +, 1, \times)$ 
    \label{th:SCsemiringhomomorphism}
  \end{theorem}
  \begin{proof}
  From theorem \ref{th:SCmonoidalhomo}, $\lambda_{\mathcal{S}(\mathcal{C})}$ is already a monoid homomorphism. We have to show that it is a homomorphism between the additive structure $(\mathcal{S}(\mathcal{C})(I,I), \cup, 0_{\mathcal{S}(\mathcal{C})}) \rightarrow (\mathbb{R}_{\geq 0}, +, 0).$ From theorem \ref{th:SCzeromorphism} we have
  \begin{equation}
    \lambda_{\mathcal{S}(\mathcal{C})}(0_{sc}) = 0.
  \end{equation} 
  Furthermore, given morphisms $\{(\gamma_i, w_i)\}_{sc}, \{(\xi_j, v_j)\}_{sc} \in \mathcal{S}(\mathcal{C})(I,I)$
  \begin{eqnarray}
    \lambda_{\mathcal{S}(\mathcal{C})}\left(\left\{(\gamma_i, w_i)\right\}_{sc} \cup \left\{(\xi_j, v_j)\right\}\right) &=& \lambda_{\mathcal{S}(\mathcal{C})}\left(\left\{(\gamma_i, w_i), (\xi_j, v_j)\right\}_{sc}\right) \notag\\
      &{=}& P_{\mathcal{S}(\mathcal{C})}\left(1_{I_{sc}}, \left\{(\gamma_i, w_i), (\xi_j, v_j)\right\}_{sc}\right) \notag\\
      &\overset{\text{Def.}}{=}& \sum_{i} w_i P_{\mathcal{Q}(\mathcal{C})}(1_{I}, \gamma_i) + \sum_{j}v_j P_{\mathcal{Q}(\mathcal{C})}(1_{I}, \xi_j) \notag\\
      &\overset{\text{Def.}}{=}& \sum_{i} w_i \lambda_{\mathcal{Q}(\mathcal{C})}(\gamma_i) + \sum_{j}v_j \lambda_{\mathcal{Q}(\mathcal{C})}(\xi_j) \notag\\
      &{=}& \lambda_{\mathcal{S}(\mathcal{C})}(\{(\gamma_i, w_i)\}_{sc}) + \lambda_{\mathcal{S}(\mathcal{C})}(\{(\xi_j, v_j)\}_{sc}).
  \end{eqnarray}

  \end{proof}

 The above theorem shows that the scalars and probabilities in summed theories have a semiring homomorphism $$P_{\mathcal{S}(\mathcal{C})}\left(\input{staterhoSC.tikz}, \input{effectsigmaSC.tikz}\right) =P_{\mathcal{S}(\mathcal{C})}\left(1_{I_{sc}}, \input{compositionrhosigmaSC.tikz}\right) = \lambda_{\mathcal{S_C}}\left(\input{compositionrhosigmaSC.tikz}\right).$$ 
 Note that for a particular theory in consideration, its summed theory may contain significant amount of redundancy. For example, in the theory $\mathcal{S}(\mathcal{Q}(\mathbf{FHilb}, |\langle -|- \rangle|^2))$, there are instances of different Kraus decompositions which are statistically equivalent e.g. $$\{(\ketbra{0}, 1/2), (\ketbra{1}, 1/2)\}_{sc}, \{(\ketbra{+}, 1/2), (\ketbra{-}, 1/2)\}_{sc}.$$ In the typical density matrix formalism of quantum theory all of these associated processes have the same density matrix or Choi matrix. However, we have not countered such a redundancy in our construction so far. We do that by quotienting again to create noisy probabilistic process theories.

  \section{Probabilities in Noisy probabilistic process theories}
  \label{sec:noisyphysicaltheories}

We formalise the final step of the series of constructions of this paper, where we perform a combination of summing and quotient. We call the resulting structure a noisy category $\mathcal{N}(\mathcal{C})$. 

\begin{definition}[\bf Noisy Category $\mathcal{N}(\mathcal{C})$]
We define the noisy category $\mathcal{N}(C)$ as $\mathcal{Q}(\mathcal{S}(\mathcal{C}))$. In other words, its morphisms are the equivalence classes of morphisms from $\mathcal{S}(\mathcal{C})$ under the probabilistic equivalence relation from Def.~\ref{def:probequivmorphisms}.
\end{definition}

By construction, since $\mathcal{N}(\mathcal{C})$ is a quotient $\mathcal{Q}$ of a probabilistic process theory $(\mathcal{S}(\mathcal{C}), \{P^A\})$, all the results from sections \ref{sec:probphysicaltheories} are true for category $\mathcal{N}(\mathcal{C})$, e.g., it is an SMC and has well-defined probability functions lifted from $(\mathcal{S}(\mathcal{C}), \{P^A\})$ that satisfy the axioms \textbf{(I)}, \textbf{(II)}, and \textbf{(III)} making it a valid simplified probabilistic process theory. Furthermore, there exists a monoid homomorphism $\lambda_{\mathcal{N}(\mathcal{C})}$ from scalars to probabilities in the noisy category $\mathcal{N}(\mathcal{C})$. 
In the following, we establish that constructing the generalised Born rule $\mathcal{G}$ and then adding noise $\mathcal{N}$ to textbook quantum mechanics, that is, constructing $\mathcal{N}(\mathcal{G}(\mathcal{U} \subseteq \mathbf{FHilb}, \{\ket{\psi}\}, \{\bra{\phi}\}, |\langle - | - \rangle|^2))  = \mathcal{N}(\mathcal{K}, \Tr[- \circ - ])$ yields a theory of completely positive maps. Similarly, adding noise to $\mathbf{FHilb}$, that is, $\mathcal{N}(\mathbf{FHilb}, \{|\langle -|- \rangle|^2\})$ returns a theory of CP maps. In other words, adding noise to textbook quantum theories returns quantum information theory.

  \begin{lemma} The statistics of morphisms is completely characterised by their action on state-effect pairs composed of single elements in their weighted sets. 
    \begin{align}
      \lambda_{\mathcal{S}(\mathcal{C})} \left( 
        \input{compositionsigmaprimefrhoprimedoubleSC.tikz}
      \right) 
      &= 
      \lambda_{\mathcal{S}(\mathcal{C})} \left( 
        \input{compositionsigmaprimefprimerhoprimedoubleSC.tikz}
      \right) \qquad\forall\, \sigma'_{sc},\, \rho'_{sc} \notag
      \\
      \iff\lambda_{\mathcal{S}(\mathcal{C})} \left( 
        \input{compositionsigmafrhodoubleSC.tikz}
      \right) 
      &= 
      \lambda_{\mathcal{S}(\mathcal{C})} \left( 
        \input{compositionsigmafprimerhodoubleSC.tikz}
      \right) \qquad\forall\, \sigma_{sc} = \{(\sigma,1)\}_{sc},\,
              \rho_{sc} = \{(\rho,1)\}_{sc}
    \end{align}
    \label{lem:characterisebyrank1}
  \end{lemma}
  \begin{proof}
    The forward implication $\implies$ is straightforward. For the converse implication $\impliedby$ consider $\{(\mathfrak{f}_i, w_i)\}_{sc} = f_{sc}$ and $\{(\mathfrak{f}'_i, w'_i)\}_{sc} = f'_{sc}$. Then, we have
    \begin{align}
      \lambda_{\mathcal{S}(\mathcal{C})} \left( 
        \input{compositionsigmafrhodoubleSC.tikz}
      \right) 
      &= 
      \lambda_{\mathcal{S}(\mathcal{C})} \left( 
        \input{compositionsigmafprimerhodoubleSC.tikz}
      \right) \qquad\forall\, \sigma_{sc} = \{(\sigma,1)\}_{sc},\,
              \rho_{sc} = \{(\rho,1)\}_{sc} \notag
      \\
      \iff\lambda_{\mathcal{S}(\mathcal{C})} \left( \left\{ \left(
        \input{compositionsigmafirhodoubleQC.tikz} , w_i \right) \right\}_{sc}
      \right) 
      &= 
      \lambda_{\mathcal{S}(\mathcal{C})} \left( \left\{ \left(
        \input{compositionsigmafiprimerhodoubleQC.tikz}, w'_i \right) \right\}_{sc}
      \right) \qquad \forall\, \sigma \in \mathcal{C}(B,I), \rho \in \mathcal{C}(I,A) \notag
      \\
      \iff\sum_i w_i\, \lambda \left(\input{compositionsigmafirhodoubleQC.tikz}
      \right)
      &=
      \sum_i w'_i\, \lambda \left(\input{compositionsigmafiprimerhodoubleQC.tikz} 
      \right) \qquad \forall\, \sigma \in \mathcal{C}(B,I), \rho \in \mathcal{C}(I,A).
      \label{eq:lemma72}
    \end{align}
 Now, consider general states and effects $\sigma'_{sc} = \{(\sigma'_{j}, u_j)\}_{sc}$ and $\rho'_{sc} = \{(\rho'_k, v_k)\}_{sc}$. For which 
  \begin{align}
      \lambda_{\mathcal{S}(\mathcal{C})} \left( 
        \input{compositionsigmaprimefrhoprimedoubleSC.tikz}
      \right) 
      &= 
      \lambda_{\mathcal{S}(\mathcal{C})} \left( 
        \input{compositionsigmaprimefprimerhoprimedoubleSC.tikz}
      \right) \qquad\forall\, \sigma'_{sc},\, \rho'_{sc} \notag \\
      \iff\lambda_{\mathcal{S}(\mathcal{C})} \left( \left\{ \left(
        \input{compositionsigmajprimefirhokprimedoubleQC.tikz} , u_j w_i v_k \right) \right\}_{sc}
      \right) 
      &= 
      \lambda_{\mathcal{S}(\mathcal{C})} \left( \left\{ \left(
        \input{compositionsigmajprimefiprimerhokprimedoubleQC.tikz}, u_j w'_i v_k \right) \right\}_{sc}
      \right) \notag\\
      &\quad \forall ~\sigma'_{sc} = \{(\sigma'_{j}, u_j)\}_{sc} \text{ and } \rho'_{sc} = \{(\rho'_k, v_k)\}_{sc}
      \notag\\
      \iff
      \sum_{j,i,k} u_j w_i v_k\, \lambda \left(
        \input{compositionsigmajprimefirhokprimedoubleQC.tikz}
      \right)
      &=
      \sum_{j,i,k} u_j w'_i v_k\, \lambda \left(
    \input{compositionsigmajprimefiprimerhokprimedoubleQC.tikz}
      \right)\notag\\
      &\quad \forall ~\sigma'_{sc} = \{(\sigma'_{j}, u_j)\}_{sc} \text{ and } \rho'_{sc} = \{(\rho'_k, v_k)\}_{sc},
    \end{align}
    but 
    \begin{align}
      \sum_i w_i\, \lambda \left(\input{compositionsigmafirhodoubleQC.tikz}
      \right)
      &=
      \sum_i w'_i\, \lambda \left(\input{compositionsigmafiprimerhodoubleQC.tikz} 
      \right) \notag\\
      & \forall\, \sigma \in \mathcal{C}(B,I), \rho \in \mathcal{C}(I,A) \notag \\
        \implies
      \sum_{j,i,k} u_j w_i v_k\, \lambda \left(
        \input{compositionsigmajprimefirhokprimedoubleQC.tikz}
      \right)
      &=
      \sum_{j,i,k} u_j w'_i v_k\, \lambda \left(
    \input{compositionsigmajprimefiprimerhokprimedoubleQC.tikz}
      \right)\\
      &\quad \forall ~\sigma'_{sc} = \{(\sigma'_{j}, u_j)\}_{sc} \text{ and } \rho'_{sc} = \{(\rho'_k, v_k)\}_{sc}. \notag
    \end{align}
    Hence, the converse implication is true. 
  \end{proof}

 This is lemma analogous to maps being characterised in terms of rank one states and effects in classical and quantum theories.
  We will now consider the construction of the category $\mathbf{CP}$. 
  Besides describing an alternative to the known $\mathbf{CP}$ construction, since the constructions apply to alternative Born rules on textbook quantum mechanics, this procedure also describes a potential method for obtaining information theories out of probabilistic process theories associated to for instance alternative quantum mechanical Born rules.

  \begin{theorem}[\bf Characterisation of equivalent states in summed $\{\mathbf{FHilb}, |\langle-|-\rangle|^2\}$] Given the original theory $\{\mathcal{C}, P^A_{\mathcal{C}}\}$ is $\{\mathbf{FHilb}, |\langle-|-\rangle|^2\}$, the states in the noisy theory correspond to unnormalised density matrices. 
  \end{theorem}
  \begin{proof}
  Consider states
  \begin{align*}
      \input{staterhoSC.tikz} = \left\{\left(\input{staterhoQCi.tikz}, w_i\right)\right\}_{sc} 
      &\sim \input{staterhoprimeSC.tikz} = \left\{\left(\input{staterhoQCprimej.tikz}, w'_j\right)\right\}_{sc}
      \\
      \overset{Thm. \ref{th:stateequiv}}{\iff}\quad 
      \lambda_{\mathcal{S}(\mathcal{C})}\left(\input{compositionrhosigmaSC.tikz}\right)
      &= 
  \lambda_{\mathcal{S}(\mathcal{C})}\left(\input{compositionrhoprimesigmaSC.tikz}\right)
      \quad \forall ~\input{effectsigmaSC.tikz}
      \\
      \overset{Lem. \ref{lem:characterisebyrank1}}{\iff}\quad 
      \lambda_{\mathcal{S}(\mathcal{C})}\left(\input{compositionrhosigmaSC.tikz}\right)
      &= 
  \lambda_{\mathcal{S}(\mathcal{C})}\left(\input{compositionrhoprimesigmaSC.tikz}\right)
      \quad \forall ~ \input{effectsigmaSC.tikz} = \left\{\left(\input{effectsigmaQC.tikz},1\right)\right\}_{sc}
      \\
      \overset{\text{Def.}}{\iff}\quad 
      \lambda_{\mathcal{S}(\mathcal{C})}\left(
        \left\{\left(\input{compositionrhoisigmaQC.tikz}, w_i\right)\right\}_{sc}
      \right)
      &= 
      \lambda_{\mathcal{S}(\mathcal{C})}\left(
        \left\{\left(\input{compositionrhoiprimesigmaQC.tikz}, w'_j\right)\right\}_{sc}
      \right) \quad \forall ~ \input{effectsigma.tikz}  \\
      \iff\quad 
      \sum_i w_i\, \lambda\left( \input{compositionrhoisigmaQC.tikz}\right)
      &= 
      \sum_j w'_j\, \lambda\left( \input{compositionrhoiprimesigmaQC.tikz}\right) \quad \forall~ \input{effectsigma.tikz}\\ 
      \iff\quad 
      \sum_i w_i\, \langle \sigma | \rho_i \rangle \langle \rho_i | \sigma \rangle
      &= 
      \sum_j w'_j\, \langle \sigma | \rho'_j \rangle \langle \rho'_j | \sigma \rangle
      \\
      \iff\quad 
    \left\langle \sigma \left| 
      \sum_i w_i\, |\rho_i\rangle \langle \rho_i| 
    \right. \right. & \left. \left.
      - \sum_j w'_j\, |\rho'_j\rangle \langle \rho'_j|
    \right| \sigma \right\rangle 
    = 0 \quad \forall\, \ket{\phi} \\
      \iff\quad 
      \sum_i w_i\, |\rho_i\rangle \langle \rho_i| 
      &= 
      \sum_j w'_j\, |\rho'_j\rangle \langle \rho'_j|.
    \end{align*}
    $\therefore$ Only those states are equivalent which have equal density matrices. Quotienting with respect to these, each state in the noisy theory corresponds to an unnormalised density matrix.
  \end{proof}

  \begin{theorem}[\bf Characterisation of equivalent morphisms in summed  $\{\mathbf{FHilb}, |\langle-|-\rangle|^2\}$] Given the original theory $\{\mathcal{C}, P^A_{\mathcal{C}}\}$ is $\{\mathbf{FHilb}, |\langle-|-\rangle|^2\}$, the morphisms in the noisy theory correspond to unnormalised positive Choi matrices i.e., completely positive maps.
    \label{thm:CPmorphisms} 
  \end{theorem}
  \begin{proof}
    Consider morphisms
    \begin{align*}
      \input{morphismfSC.tikz} = \left\{\left(\input{morphismfiQC.tikz}, w_i\right)\right\}_{sc} 
      &\sim 
      \input{morphismfprimeSC.tikz} = \left\{\left(\input{morphismfprimeiQC.tikz}, w'_i\right)\right\}_{sc} 
      \\
      \overset{\text{Def.}}{\iff}\quad 
      \lambda_{\mathcal{S}(\mathcal{C})} \left(
        \input{compositionsigmaprimefrhoprimedoubleSC.tikz}
      \right)
      &=
      \lambda_{\mathcal{S}(\mathcal{C})} \left(
        \input{compositionsigmaprimefprimerhoprimedoubleSC.tikz}
      \right)
      \\
      &\quad \forall\, \sigma'_{sc},\, \rho'_{sc}
      \\
      \overset{Lem.\ref{lem:characterisebyrank1}}{\iff}\quad 
      \lambda_{\mathcal{S}(\mathcal{C})} \left(
        \input{compositionsigmafrhodoubleSC.tikz}
      \right)
      &=
      \lambda_{\mathcal{S}(\mathcal{C})} \left(
        \input{compositionsigmafprimerhodoubleSC.tikz}
      \right)
      \\
      &\quad \forall\, \sigma_{sc} = \{(\sigma,1)\}_{sc},\,
      \rho_{sc} = \{(\rho,1)\}_{sc}
      \\
      \iff\quad 
      \lambda_{\mathcal{S}(\mathcal{C})} \left(
        \left\{\left(\input{compositionsigmafirhodoubleQC.tikz}, w_i\right)\right\}_{sc}
      \right)
      &=
      \lambda_{\mathcal{S}(\mathcal{C})} \left(
        \left\{\left(\input{compositionsigmafiprimerhodoubleQC.tikz}, w'_i\right)\right\}_{sc}
      \right)
      \\
      &\quad \forall\, \sigma,\, \rho
      \\
      \iff\quad 
      \sum_i w_i\, \lambda\left(
        \input{compositionsigmafirhodoubleQC.tikz}
      \right)
      &=
      \sum_i w'_i\,  \lambda\left(
        \input{compositionsigmafiprimerhodoubleQC.tikz}
      \right)\\
      \iff\quad 
      \sum_i w_i\, \langle \sigma | (\mathfrak{f}_i \otimes 1_A) | \rho \rangle 
      \langle \rho | (\mathfrak{f}_i^\dagger \otimes 1_A) | \sigma \rangle
      &=
      \sum_i w'_i\, \langle \sigma | (\mathfrak{f}'_i \otimes 1_A) | \rho \rangle 
      \langle \rho | ({\mathfrak{f}'}_i^\dagger \otimes 1_A) | \sigma \rangle
      \\
      \iff\quad 
      \left\langle \sigma \middle| 
        \sum_i w_i\, (\mathfrak{f}_i \otimes 1_A) | \rho \rangle \langle \rho | (\mathfrak{f}_i^\dagger \otimes 1_A)
        \right. &- \left.
        \sum_i w'_i\, (\mathfrak{f}'_i \otimes 1_A) | \rho \rangle \langle \rho | ({\mathfrak{f}'}_i^\dagger \otimes 1_A)
      \middle| \sigma \right\rangle
      = 0
      \\
      \iff\quad 
      \sum_i (\sqrt{w_i} \mathfrak{f}_i \otimes 1_A) | \rho \rangle \langle \rho | (\sqrt{w_i} \mathfrak{f}_i^\dagger \otimes 1_A)
      &=
      \sum_i (\sqrt{w'_i} \mathfrak{f}'_i \otimes 1_A) | \rho \rangle \langle \rho | (\sqrt{w'_i} {\mathfrak{f}'}_i^\dagger \otimes 1_A)
      \\
      &\quad \forall\, \ket{\rho}
      \\
      \iff\quad 
      \{\sqrt{w_i} \mathfrak{f}_i\} \text{ and } \{\sqrt{w'_i} \mathfrak{f}'_i\}
      &\text{ have the same Choi matrix.}
    \end{align*}

    $\therefore$ Only those morphisms are equivalent which have equal Choi matrices.
    Furthermore, the Choi matrices are positive since 
    \[
      C = \sum_i (\sqrt{w_i}\mathfrak{f}_i \otimes 1_A)\ket{n}\rangle\rangle\langle\langle n|(\sqrt{w_i}\mathfrak{f}_i^\dagger \otimes 1_A) = M M^\dagger,
    \]
    where 
    \[
      M = \sum_i (\sqrt{w_i}\mathfrak{f}_i \otimes 1_A)| n \rangle\rangle\langle\langle n|/\sqrt{d}, 
    \]
    $d$ is the dimension of the output space, and 
    $\ket{n}\rangle = \sum_j \ket{j} \ket{j}$ is the unnormalised maximally entangled state.  By Choi’s theorem \cite{Choi1975Jun}, the matrices correspond to completely positive maps.
  \end{proof}

  \begin{theorem}[\bf Characterisation of probability functions in noisy $\{\mathbf{FHilb}, |\langle-|-\rangle|^2\}$]
    The probability functions in a noisy probabilistic process theory constructed from  the probabilistic process theory $\{\mathbf{FHilb}, |\langle - | - \rangle|^2\}$  is given by $P_{\mathcal{N}(\mathcal{C})}(\rho_{nc}, \sigma_{nc}) = \Tr[\rho \sigma]$, where $\rho$ and $\sigma$ are the corresponding density matrices.
  \end{theorem}
  \begin{proof}
  Let  $(\rho_i, w_i) \in \rho_{sc} \in \rho_{nc}$. Similarly, $(\sigma_i, v_i) \in \sigma_{sc} \in \sigma_{nc}.$ Then, $P_{\mathcal{N}(\mathcal{C})}(\rho_{nc}, \sigma_{nc}) = P_{\mathcal{S}(\mathcal{C})}(\rho_{sc}, \sigma_{sc}) = \sum_{i,j} w_i v_j P(\rho_i, \sigma_j)  = \sum_{i,j} w_i v_j |\langle \sigma_j | \rho_i \rangle |^2$. If we define density matrices for $\rho_{nc}$, and $\sigma_{nc}$ as $ \rho = \sum_i w_i \ket{\rho_i}\bra{\rho_i}$ and $\sigma = \sum_j v_j \ket{\sigma_j}\bra{\sigma_j}$ respectively then, $\sum_{i,j} w_i v_j |\langle \sigma_j | \rho_i \rangle |^2 =  \Tr(\rho \sigma)$ and hence the desired relation is derived. 
  \end{proof}

  \begin{corollary}
    For the SPPT $\{\mathbf{FHilb}, |\langle - | - \rangle|^2\}$, The corresponding noisy probabilistic process theory is $\{\mathbf{CP}, \Tr(- \circ -)\}$. 
  \end{corollary}

  Thus, we have achieved a construction of the category $\mathbf{CP}$ of completely positive maps from simple constraints on probability functions which do not utilise the adjoint structure or notion of discard. The characterisation of $\mathcal{N}(\mathcal{K}, \Tr[-\circ -])$ follows identically to Thm. \ref{thm:CPmorphisms} with the only difference being $\mathfrak{f}_i$ are now rank-1 CPTNI maps instead of rank-1 CP maps. However, since we are allowed to have arbitrary positive weights $w_i$, the resulting morphisms are still general CP maps.

  \begin{corollary}
    For the SPPT $\{\mathcal{K}, Tr[- \circ -]\}$, The corresponding noisy probabilistic process theory is $\{\mathbf{CP}, \Tr(- \circ -)\}$. 
  \end{corollary}

  Before proceeding with the final result of this section: the existence of Born rule in noisy probabilistic process theories, we need to borrow some additional structure from the summed theory.

  \begin{definition}[\bf Additive union in noisy theory]
    A noisy category $\mathcal{N}(\mathcal{C})$ inherits the additive union from the summed theory $\mathcal{S}(\mathcal{C})$. For morphisms $[f]_{nc} \text{ and } [g]_{nc} \in \mathcal{N}(\mathcal{C})(A,B)$, where $[f]_{nc} \text{ and } [g]_{nc}$ are equivalence classes of morphisms $f \text{ and } g \in \mathcal{S}(\mathcal{C})(A,B)$ respectively, we define $[f]_{nc} \cup_{nc} [g]_{nc} := [f \cup_{sc} g]_{nc}$. 
  \end{definition}

  \begin{lemma}[\bf Consistency of the additive union]
    The definition of additive union in $\mathcal{N}(\mathcal{C})$ is well-formed, that is  $f'_{sc} \sim f_{sc}$ and $g'_{sc} \sim g_{sc}$ implies $f_{sc} \cup_{sc} g_{sc} \sim f'_{sc} \cup_{sc} g'_{sc}$. 
  \end{lemma}
  \begin{proof}
    Assuming $f'_{sc} \sim f_{sc}$ and $g'_{sc} \sim g_{sc},$ 
    \begin{align}
      \lambda_{\mathcal{S}(\mathcal{C})}\left(\input{compositionsigmafuniongrhodoubleSC.tikz}\right) = \sum_{i,j,k} w_i v_j u_k \lambda\left(\input{compositionsigmakfirhojdoubleQC.tikz}\right) + \sum_{i,j,k} t_i v_j u_k \lambda\left(\input{compositionsigmakgirhojdoubleQC.tikz}\right) \notag
    \end{align}
    \begin{align}
      &= \sum_{i,j,k} w_i v_j u_k \lambda\left(\input{compositionsigmakfprimeirhojdoubleQC.tikz}\right) + \sum_{i,j,k} t_i v_j u_k \lambda\left(\input{compositionsigmakgprimeirhojdoubleQC.tikz}\right) \notag\\
      &=  \lambda_{\mathcal{S}(\mathcal{C})}\left(\input{compositionsigmafprimeuniongprimerhodoubleSC.tikz}\right)  \quad \forall \rho_{sc}, \sigma_{sc} \notag\\
      & \implies f_{sc} \cup g_{sc} \sim f'_{sc} \cup g'_{sc}.
    \end{align}
  \end{proof}

  \begin{lemma}[\bf Identity of additive union]
    The zero morphism $0_{nc}$, defined as the equivalence class of $0_{sc}$ is the unit for additive union $\cup_{nc}$ in $\mathcal{N}(\mathcal{C})$. 
  \end{lemma}
  \begin{proof}
    Let $f_{nc}$ be a morphism in the noisy category with $f_{sc} \in f_{nc} = [f_{sc}]_{nc}$, then, $f_{nc} \cup_{nc} 0_{nc} = [f_{sc} \cup 0_{sc}]_{nc} = [f_{sc}]_{nc} = f_{nc}$. Similarly, $0_{nc} \cup_{nc}  f_{nc} = [0_{sc} \cup f_{sc}]_{nc} = [f_{sc}]_{nc} = f_{nc}$.
  \end{proof}

  We are now ready to establish the Born rule in noisy probabilistic process theories. We show that the scalars of a noisy theory are isomorphic to the probability range with semiring isomorphisms $\lambda_{\mathcal{N}(\mathcal{C})}$ and $\theta_{\mathcal{N}(\mathcal{C})}$.

  \begin{theorem}[\bf Born rule as semiring isomorphism in noisy theories]
  For every probabilistic process theory there exists an associated simplified probabilistic process theory with a generalised Born rule that is additive.
  \end{theorem}
  \begin{proof}
  Concretely, we show that there exists a semiring isomorphism $\lambda_{\mathcal{N}(\mathcal{C})}$ between the scalars $(\mathcal{N}(\mathcal{C})(I, I), 1_I, \otimes_{nc}, 0_{nc}, \cup_{nc})$ and the range of probabilities $(Range(P_{\mathcal{N}(\mathcal{C})}) \subseteq \mathbb{R}_{\geq 0}, 1, \times, 0, + )$.  
    Since $\mathcal{N}(\mathcal{C})$ is a quotiented theory, there exist monoid homomorphism $\lambda_{\mathcal{N}(\mathcal{C})}$ from $(\mathcal{N}(\mathcal{C})(I, I), 1_I, \otimes_{nc})$ to $(Range(P_{\mathcal{N}(\mathcal{C})}) \subseteq \mathbb{R}_{\geq 0}, 1, \times)$. This homomorphism can be lifted to a semiring homomorphism. From Theorem \ref{th:SCsemiringhomomorphism} we have
    $$\lambda_{\mathcal{N}(\mathcal{C})}(0_{nc}) = \lambda_{\mathcal{S}(\mathcal{C})}(0_{sc}) = 0.$$ 
  Let $\gamma_{nc} = [\{(\gamma_i, w_i)\}_{sc}]_{nc}, \text{ and } \xi_{nc} = [\{(\xi_j, v_j)\}_{sc}]_{nc}.$ We have
  \begin{eqnarray}
    \lambda_{\mathcal{N}(\mathcal{C})}\left(\gamma_{nc} \cup_{nc} \xi_{nc}\right) &=& \lambda_{\mathcal{S}(\mathcal{C})}\left(\{(\gamma_i, w_i)\}_{sc}\right) + \lambda_{\mathcal{S}(\mathcal{C})}\left(\{(\xi_j, v_j)\}_{sc}\right) = \lambda_{\mathcal{N}(\mathcal{C})}\left(\gamma_{nc}\right) +  \lambda_{\mathcal{N}(\mathcal{C})}\left(\xi_{nc}\right)
  \end{eqnarray}
  thus, $\lambda_{\mathcal{N}(\mathcal{C})}$ is a semiring homomorphism. Now, let $\omega_{sc} \in \omega_{nc}, ~\omega'_{sc} \in \omega'_{nc}$ be scalars in $\mathcal{N}(\mathcal{C})$ so that
    $$\lambda_{\mathcal{N}(\mathcal{C})}(\omega_{nc}) = \lambda_{\mathcal{N}(\mathcal{C})}(\omega'_{nc}) \implies \lambda_{\mathcal{S}(\mathcal{C})}(\omega_{sc}) = \lambda_{\mathcal{S}(\mathcal{C})}(\omega'_{sc}).$$
  Then, by Theorem \ref{th:scalarequiv}, we have $\omega_{sc} \sim \omega'_{sc}$ and hence $\omega_{nc} = [\omega_{sc}]_{nc} = [\omega'_{sc}]_{nc} = \omega'_{nc} $. Thus, $\lambda_{\mathcal{N}(\mathcal{C})}$ is one-one. Therefore, we can define $\theta_{\mathcal{N}(\mathcal{C})}: Range(P_{\mathcal{N}(\mathcal{C})}) \rightarrow \mathcal{N}(\mathcal{C})(I, I)$ as $\theta_{\mathcal{N}(\mathcal{C})}(p) = \lambda^{-1}_{\mathcal{N}(\mathcal{C})}(p)$, where $p \in Range(P_{\mathcal{N}(\mathcal{C})})$. Then, $\lambda_{\mathcal{N}(\mathcal{C})}(\theta_{\mathcal{N}(\mathcal{C})}(p)) = p$ and $\theta_{\mathcal{N}(\mathcal{C})}(\lambda_{\mathcal{N}(\mathcal{C})}(\omega_{nc})) = \omega_{nc}$ for all $\omega_{nc} \in \mathcal{N}(\mathcal{C})(I, I)$. Since $\lambda_{\mathcal{N}(\mathcal{C})}$ is bijective in the Range of probabilities, the inverse $\theta_{\mathcal{N}(\mathcal{C})}$ is also a semiring homomorphism.
  \end{proof}


The result of adding noise, presents the strongest form of similarity between the structures of its scalars and probabilities, since, semiring homomorphisms constitute a highly constrained subset of monoid homomorphisms. This is reflected in quantum theory: the Born rule in the noisy quantum theory ($\Tr[\rho \sigma]$ a composition of state and effect) is considerably closer to the generalised Born rule than in the case of pure quantum theory ($|\langle \phi | \psi \rangle|^2$ modulus square of the composition). This happens because a monoid homomorphism from the scalars of $\mathbf{FHilb}$, that is $\mathbb{C}$ to the probability range $\mathbb{R}_{\geq 0}$ allows for squaring or any other positive powers to obtain probabilities $$\lambda(re^{i \theta}) = r^k, \qquad   re^{i \theta} = \langle \phi | \psi \rangle,$$
    whereas a semiring isomorphism would not allow for such powers (they are forbidden by the preservation of sums). The only semiring isomorphism from the scalars of $\mathbf{CP}$, that is $\mathbb{R}_{\geq 0}$ to probability range $\mathbb{R}_{\geq 0}$ for instance, is an identity function.

    $$\lambda_{\mathcal{N}(\mathcal{C})}(r) = r, \qquad r = \Tr[\rho \sigma].$$ 

\begin{lemma}[\bf Rigidity of the Semiring of Positive Reals]
  Let $\lambda: \mathbb{R}_{\geq 0} \to \mathbb{R}_{\geq 0}$ be a semiring homomorphism. That is, $\lambda(0)=0$, $\lambda(1)=1$, $\lambda(a+b) = \lambda(a)+\lambda(b)$, and $\lambda(ab) = \lambda(a)\lambda(b)$. Then $\lambda(x) = x$ for all $x \in \mathbb{R}_{\geq 0}$.
\end{lemma}

\begin{proof}
  First, observe that $\lambda(n) = n$ for any natural number $n \in \mathbb{N}$. This follows by induction from $\lambda(1)=1$ and additivity: $\lambda(n+1) = \lambda(n) + \lambda(1) = n + 1$.
  
  Second, for any positive rational number $q = \frac{m}{n}$ where $m, n \in \mathbb{N}_{>0}$, we have $n \cdot q = m$. Applying $\lambda$ to both sides yields $\lambda(n \cdot q) = \lambda(m)$. Using the multiplicative property and the result for naturals, we have $\lambda(n)\lambda(q) = n\lambda(q) = m$. This implies $\lambda(q) = \frac{m}{n} = q$. Thus, $\lambda$ is the identity on the rationals $\mathbb{Q}_{\geq 0}$.
  
  Third, we show that $\lambda$ is order-preserving. Let $x \leq y$. Then $y = x + \delta$ for some $\delta \in \mathbb{R}_{\geq 0}$. Applying $\lambda$, we get $\lambda(y) = \lambda(x) + \lambda(\delta)$. Since the codomain is $\mathbb{R}_{\geq 0}$, we must have $\lambda(\delta) \geq 0$, and therefore $\lambda(x) \leq \lambda(y)$.
  
  Finally, for any real number $x \in \mathbb{R}_{\geq 0}$, there exist sequences of rationals $\{q_n\}$ and $\{p_n\}$ such that $q_n \leq x \leq p_n$ with $\lim_{n \to \infty} q_n = x$ and $\lim_{n \to \infty} p_n = x$. By the order-preserving property, we have $\lambda(q_n) \leq \lambda(x) \leq \lambda(p_n)$. Since $\lambda$ fixes rationals, this becomes $q_n \leq \lambda(x) \leq p_n$. Because the real numbers are complete, the only number that is greater than or equal to every rational lower bound of $x$ and less than or equal to every rational upper bound of $x$ is $x$ itself. Therefore, $\lambda(x) = x$.
  
\end{proof}

   This leads to following corollary: 
    
      \begin{corollary}
      In any probabilistic physical theory arising from the construction $\mathcal{N}$, with the underlying SMC being $\mathbf{CP}$, the possible probability functions on that theory are constrained to $P(\rho, \sigma) = \lambda(\sigma \circ \rho) = \Tr[\rho \sigma].$ 
    \end{corollary}
    In other words, any probabilistic process theory that produces $\mathbf{CP}$ as its noisy theory must also produce the quantum mechanical Born rule as its generalised Born rule. The relative abundance of monoid homomorphisms as compared to semiring homomorphisms suggests that on any probabilistic process theory a more robust equivalence would exist between probabilities and scalars in its information theory as compared to its initial formulation.

  \section{Conclusion}
  \label{sec:conclusion}

In this work, we have formalised and derived the generalised Born rule in process theories by introducing probabilistic process theories. These consist of compositional structures (categories) from which physical processes, states, and effects are carved out, along with statistical structures (probability functions) and their constraints: associativity with respect to transformations, factorisation of probabilities for independent transformations and non-triviality.  A central feature of such theories is their ability to accommodate physically motivated structures such as textbook quantum mechanics which does not strictly form a category in the most naive sense. We consider their simplified forms in simplified probabilistic process theories and show that such theories can always be transformed via quotienting $\mathcal{Q}$ to process theories which have a Born rule.  Subsequently we demonstrate that a similar construction holds for the general case; any probabilistic process theory theory can be systematically completed into a simplified probabilistic process theory with a generalised Born rule via a quotienting procedure $\mathcal{G}$. As a consequence, we have established that the generalised Born rule is a derivable feature of any reasonable operational theory.

After demonstrating that the axioms of a probabilistic process theory give rise to a homomorphism between state-effect compositions and their observed probabilities, this relationship is progressively strengthened through the addition of noise. In the quotiented theory $(\mathcal{Q}(\mathcal{C}),\{P^A\})$, this correspondence is a monoid isomorphism, which can in principle take a multitude of forms (e.g., $|\langle- | -\rangle|^k$). However, upon introducing noise via the summed $\mathcal{S}(\mathcal{C})$ and noisy categories $\mathcal{N}(\mathcal{C})$, this relationship elevates to a semiring isomorphism. This establishes that for a broad class of noisy theories, the scalar value obtained from composing a state and an effect is, in a structurally identical sense, the probability of the measurement outcome. 

 We also show that the category of completely positive maps, $\mathbf{CP}$, arises naturally via two distinct but convergent paths. 
  On one hand, applying our procedure $\mathcal{N}$ to the simplified probabilistic process theory of finite-dimensional Hilbert spaces, $\mathbf{FHilb}$, yields $\mathbf{CP}$ directly.  On the other hand, and perhaps less straightforwardly,, applying the quotienting procedure $\mathcal{G}$ to the probabilistic process theory of textbook quantum mechanics (consisting only of pure states, effects and unitaries) first yields the theory of Kraus operators (rank-1 CPTNI maps), which upon the introduction of noise also lifts to $\mathbf{CP}$. 
The construction of completely positive maps given is entirely adjoint-free, in contrast to existing approaches such as the $\mathbf{CPM}$  and $\mathbf{CP}$-infinity constructions. 
This adjoint-free nature of the construction suggests that it might be possible to combine with the results of \cite{Coecke_2018} (which take semi-additive categories as a starting point), in order to add classical objects to general probabilistic process theories.

A consequence of not accounting for any adjoint structure is the lack of a built-in normalisation, or more precisely, determinism. In probabilistic process theories there seem to be natural candidates on account of the existence from the beginning of a privileged collection of physical processes within the pure theory. Alternatively, one can naturally imagine not only identifying the legitimate states and measurement outcomes, but further identifying the sets of measurement outcomes which together constitute measurements. The convergence of the two approaches likely depends on the existence of stinespring-like representation theorems for the theory at hand. 

Managing to refine the constructions of this paper in such a way as to identify the deterministic processes and measurements is a particularly promising route to future applications of this derivation of the generalised Born rule. First, such a refinement would provide a way to construct operational probabilistic theories (OPTs) from probabilistic process theories, and so, give a general method for constructing OPTs. Most notably, since probabilistic process theories can neatly cope with alternative quantum mechanical Born rules, this would allow for a construction from alternative quantum mechanical Born rules to operational probabilistic theories. Such a construction could in turn give a new perspective on reconstructions of the amplitude-square Born rule of quantum theory, in particular, by observing properties of the OPTs induced by alternative Born rules one could systematically rule them out by those observed features that might be deemed to be pathological \cite{Bao2015Nov,Aaronson2004Jan, Galley2018Nov, Masanes2019Mar, Alegre2025Dec}.



Finally, the construction of an \textit{equivalent} theory with more desirable properties is a familiar story in category theory, notably, a justification for process theories exists already from the strictification theorem for monoidal categories, in which every monoidal category, where for instance associativity of the tensor product holds only up to coherent isomorphism, is proven equivalent as a monoidal category to one in which associativity holds on the nose. 
A natural question then, is whether general non-strict monoidal categories equipped with probability functions can be proven suitably equivalent to strict symmetric monoidal categories with generalised Born rules for their probabilities. 
To give such a simultaneous strictification for symmetric monoidal and probabilistic structure would require a careful reframing of probabilistic process theories to the non-strict case, and furthermore, a careful and precise statement on what exactly it should mean formally to say that two physical theories (symmetric monoidal categories equipped with probability functions) are equivalent. This in turn naturally motivates the construction of categories of probabilistic process theories, their non-strict generalisations, and the functors and natural transformations between them. 
In short, the framing of probabilistic process theories and the generalised Born rule represents an outline for a potential general justification and strictification theorem for modern approaches to quantum foundations.

  \bibliographystyle{quantum}
  \bibliography{generic}

\appendix

  \end{document}